\documentclass[11pt]{article}

\usepackage{microtype}
\usepackage{epsfig}
\usepackage{microtype}
\usepackage{graphicx}
\usepackage{enumitem}
\usepackage[linesnumbered,ruled, lined]{algorithm2e}
\usepackage{algpseudocode}
\usepackage{epsfig,enumerate,amsmath,amsfonts,amssymb,amsthm,mathrsfs,ifpdf}
\usepackage{indentfirst,relsize}
\usepackage[numbers]{natbib}
\usepackage{setspace}
\usepackage{enumerate}
\usepackage{latexsym}
\usepackage{stackrel}
\usepackage[all]{xy}
\usepackage[margin = 2.5cm]{geometry}
\usepackage[usenames,dvipsnames]{pstricks}
\usepackage{pst-grad} 
\usepackage{pst-plot} 
\usepackage{xspace}
\usepackage{hyperref}

\newcommand{\size}[1]{\left| #1 \right|}

\newcommand{\E}{\mathbb{E}}
\newcommand{\V}{\mathbb{V}}
\newcommand{\remove}[1]{}

\newcommand{\R}{\mathbb{R}}
\newcommand{\N}{\mathbb{N}}
\newcommand{\cN}{\mathcal{N}}

\newcommand{\cT}{\mathcal{T}}

\newcommand{\cE}{\mathcal{E}}

\newcommand{\cAe}{\mathcal{A}_{\text{exact}}}
\newcommand{\cAc}{\mathcal{A}_{\text{coarse}}}

\newcommand{\cB}{\mathcal{B}}
\newcommand{\cD}{\mathcal{D}}
\newcommand{\cP}{\mathcal{P}}

\newcommand{\cR}{\mathcal{R}}
\newcommand{\cG}{\mathcal{G}}

\newcommand{\cO}{\mathcal{O}}
\newcommand{\Oh}{\mathcal{O}}
\newcommand{\tOh}{\tilde{\mathcal{O}}}
\newcommand{\cF}{\mathcal{F}}

\newcommand{\cH}{\mathcal{H}}

\newcommand{\cZ}{\mathcal{Z}}
\newcommand{\eps}{\epsilon}

\newcommand{\complain}[1]{\textcolor{red}{#1}}
\newcommand{\comments}[1]{\textcolor{blue}{\bf{#1}}}

\newcommand{\dsubset}{A_1,\ldots, A_d}
\newcommand{\ssubset}{A_1^{[a_1]},\ldots, A_s^{[a_s]}}
\newcommand{\tsubset}{B_1^{[b_1]},\ldots, B_t^{[b_t]}}

\newcommand{\gpis}{{\sc GPIS}\xspace}
\newcommand{\gonepis}{$\mbox{\gpis}_1$ }
\newcommand{\gtwopis}{$\mbox{\gpis}_2$ }
\newcommand{\bis}{{\sc BIS}\xspace}
\newcommand{\tis}{{\sc TIS}\xspace}
\newcommand{\colored}{{\sc COLORED}\xspace}
\newcommand{\hest}{{\sc Hyperedge-Estimation}\xspace}

\newcommand{\test}{{\sc Triangle-Estimation}\xspace}
\newcommand{\pr}{{\mathbb{P}}\xspace}
\newcommand{\verest}{{\sc Verify-Estimate}\xspace}
\newcommand{\cest}{{\sc Coarse-Estimate}\xspace}
\newcommand{\est}{\mbox{{\sc Est}}}
\newcommand{\act}{\mbox{{\sc Act}}}
\newcommand{\tuple}{\mbox{{\sc Tuple}}}
\newcommand{\lbeps}{\left({n^{-d}\log ^{5d+5} n} \right)^{1/4}}

\newcommand{\etal}{{\it{et al.}}}

\theoremstyle{plain}
\newtheorem{theo}{Theorem}[section]
\newtheorem{lem}[theo]{Lemma}
\newtheorem{pro}[theo]{Proposition}

\newtheorem{cl}[theo]{Claim}

\theoremstyle{definition}
\newtheorem{defi}[theo]{Definition}

\newtheorem{obs}[theo]{Observation}
\newtheorem{fact}[theo]{Fact}

\newcommand{\defproblem}[3]{
  \vspace{1mm}
\noindent\fbox{
  \begin{minipage}{0.96\textwidth}
  \begin{tabular*}{\textwidth}{@{\extracolsep{\fill}}lr} #1 \\ \end{tabular*}
  {\bf{Input:}} #2  \\
  {\bf{Output:}} #3
  \end{minipage}
  }
  \vspace{1mm}
}

\begin{document}
\title{Hyperedge Estimation using Polylogarithmic Subset Queries}

\author{
	Anup Bhattacharya\thanks{Supported by NPDF Fellowship, India, ACM Unit, Indian Statistical Institute, Kolkata, India}
\and Arijit Bishnu \footnote{
ACM Unit, Indian Statistical Institute, Kolkata, India.
}
\and
Arijit Ghosh
\footnotemark[2]
\and
Gopinath Mishra
\footnotemark[2]
}

\date{}
\maketitle

\maketitle
\begin{abstract}

\noindent In this work, we estimate the number of hyperedges in a hypergraph $\cH(U(\cH),\cF(\cH))$, where $U(\cH)$ denotes the set of vertices and $\cF(\cH)$ denotes the set of hyperedges. We assume a query oracle access to the hypergraph $\cH$. Estimating the number of edges, triangles or small subgraphs in a graph is a well studied problem. Beame \etal~and Bhattacharya \etal~gave algorithms to estimate the number of edges and triangles in a graph using queries to the {\sc Bipartite Independent Set} ({\sc BIS}) and the {\sc Tripartite Independent Set} ({\sc TIS}) oracles, respectively. We generalize the earlier works by estimating the number of hyperedges using a query oracle, known as the {\bf Generalized $d$-partite independent set oracle (\gpis)}, that takes $d$ (non-empty) pairwise disjoint subsets of vertices $\dsubset \subseteq U(\cH)$ as input, and answers whether there exists a hyperedge in $\cH$ having (exactly) one vertex in each $A_i, i \in \{1,2,\ldots,d\}$. \remove{Our work on estimating the number of hyperedges in a hypergraph using query oracles can be seen as a natural generalization of these above works.} We give a randomized algorithm for the hyperedge estimation problem using the \gpis query oracle to output $\widehat{m}$ for $m(\cH)$ satisfying $(1-\eps) \cdot m(\cH) \leq \widehat{m} \leq (1+\eps) \cdot m(\cH)$. The number of queries made by our algorithm, assuming $d$ to be a constant, is polylogarithmic in the number of vertices of the hypergraph.

\end{abstract}

\section{Introduction} \label{sec:intro}

\noindent A hypergraph $\cH$ is a \emph{set system} $(U(\cH),\cF(\cH))$, where $U(\cH)$ denotes a set of $n$ vertices and $\cF(\cH)$, a set of subsets of $U(\cH)$, denotes the set of hyperedges. A hypergraph $\cH$ is said to be $d$-uniform if every hyperedge in $\cH$ consists of exactly $d$ vertices. The cardinality of the hyperedge set is denoted as $m(\cH)=\size{\cF(\cH)}$. We investigate the problem of estimating $m(\cH)$, where the hypergraph $\cH$ can be accessed only by queries to an oracle that answers subset queries of a particular kind.

In the \emph{subset size estimation} problem using the \emph{query model} of computation, the \emph{subset query} oracle is used to estimate the size of an unknown subset $S \subseteq U$, where $U$ is a known universe of elements. A \emph{subset query} with a subset $T \subseteq U$ asks whether $S \cap T$ is empty or not. \remove{The estimation technique involves querying the subset query oracle with different subsets $T$ and use some concentration bounds to estimate the size of $S$.} Viewed differently, this is about estimating an unknown set $S$ by looking at its intersection pattern with a known set $T$. At its core, a subset query essentially enquires about the existence of an intersection between two sets -- a set chosen by the algorithm designer and an unknown set whose size we want to estimate. \remove{A \emph{set membership query} is a special case of a subset query, where one asks a Yes/No question about the existence of an element in a set.} The study of subset queries was initiated in a breakthrough work by Stockmeyer~\cite{Stockmeyer83,Stockmeyer85} and later formalized by Ron and Tsur~\cite{RonT16}.

\remove{
\comments{In subset size estimation problem using subset query oracle, we want to estimate the size of a subset $S\subseteq T$ using an oracle which tells us, given a query set $Q\subseteq T$, whether $Q \cap T$ is empty or not. Viewed differently, this is about estimating an unknown set $S$ by looking at its intersection pattern with a known set $T$.}
}

Beame \etal~\cite{BeameHRRS18} used a subset query oracle, named {\sc Bipartite Independent Set} (\bis) query oracle to estimate the number of edges in a graph using polylogarithmic query complexity\footnote{query complexity means the number of queries used by the corresponding query oracle}. The \bis query oracle answers a YES/NO question on the existence of an edge between two disjoint subsets of vertices of a graph $G$. Having estimated the number of edges in a graph using \bis{} queries, a very natural question was to estimate the number of hyperedges in a hypergraph using an appropriate query oracle. The answer to the above question is not obvious as two edges in a graph can intersect in at most one vertex but the intersection between two hyperedges in a hypergraph can be an arbitrary set. As a first step towards resolving this generalized question, Bhattacharya et al.~\cite{Bhatta-abs-1808-00691, BhattaISAAC} considered the hyperedge estimation problem using a {\sc Tripartite Independent Set} ({\sc TIS}) oracle in $3$-uniform hypergraphs. A {\sc TIS} query oracle takes three disjoint subsets of vertices as input and reports whether there exists a hyperedge having a vertex in each of the three sets. They showed that when the number of hyperedges having two vertices in common is bounded above (polylogarithmic in $n$), then the number of hyperedges in a $3$-uniform hypergraph can be estimated using polylogarithmic TIS queries. This leads us to ask the next set of questions given as follows.
\begin{itemize}
        \item \emph{{\bf Question 1}: For a $3$-uniform hypergraph, is the dependence of the {\sc TIS} query complexity on the number of hyperedges with two common vertices inherent as in Bhattacharya et al.~\cite{Bhatta-abs-1808-00691, BhattaISAAC}?}
        \item \emph{{\bf Question 2}: Can the subset query oracle framework of Beame \etal~be extended to estimate the number of hyperedges in a $d$-uniform hypergraph using only polylogarithmic many queries?}
\end{itemize}

In this paper, we give positive answers to both these questions. We show that the number of hyperedges in a $d$-uniform hypergraph can be estimated using polylogarithmic~\footnote{Here the exponent of $\log n$ is $\Oh(d)$, but the exponent of $\eps$ is an absolute constant.} many \gpis queries. \remove{Even though our algorithmic framework is similar to that of Beame et al.~\cite{BeameHRRS18} and Bhattacharya et al.~\cite{Bhatta-abs-1808-00691, BhattaISAAC}, the techniques used in this paper are different.} Next, we define formally our query model and state the main results.

\subsection{Query model, problem description and our results} \label{ssec:querymodel}
\noindent We start by defining our query oracle. 
\begin{defi}\label{def:gpis} 
{Generalized $d$-partite independent set oracle (\gpis)~\cite{BishnuGKM018}:} Given $d$ pairwise disjoint subsets of vertices $\dsubset \subseteq U(\cH)$ of a hypergraph $\cH$ as input, \gpis{} query oracle answers {\sc Yes} if and only if $m(\dsubset) \neq 0$, where $m(\dsubset)$ denotes the number of hyperedges in $\cH$ having exactly one vertex in each $A_i$, $\forall i \in \{1,2,\ldots,d\}$. \end{defi}

Observe that the \gpis query is a generalization of subset queries, as for $d=1$, \gpis is equivalent to asking a Yes/No question about the existence of any element of interest in the queried set. For $d=2$, \gpis{} is \bis{}. An involved use of an induction on $d$ will show how \gpis generalizes from the subset queries and this process will unravel the intricate intersection pattern of the $d$-uniform hyperedges. \bis{}, \tis{} and their generalization, the \gpis{} query has a \emph{transversal} nature to it. A \emph{transversal}~\cite{Matgeo} of a hypergraph $\cH=(U(\cH), \cF(\cH))$ is a subset $T \subseteq U(\cH)$ that intersects all sets of $\cF(\cH)$. One can see a \gpis{} query as a \emph{transversal query} as it answers if there exists a transversal for the disjoint subsets $\dsubset$ as in Definition~\ref{def:gpis}.

We now state the precise problem that we solve in the \gpis oracle framework and present our main result in Theorem~\ref{theo:main_restate}.

\defproblem{\hest}{A set of $n$ vertices $U(\cH)$ of a hypergraph $\cH$, a \gpis{} oracle access to $\cH$, and $\eps\in (0,1)$.}{An estimate $\widehat{m}$ of $m(\cH)$ such that $(1 - \epsilon) \cdot m(\cH) \leq \widehat{m} \leq (1 - \epsilon) \cdot m(\cH)$.}

\begin{theo} \label{theo:main_restate} Let $\cH$ be a hypergraph with $\size{U(\cH)}=n$. For any $\eps \in (0,1)$, \hest can be solved using ${\cO}_d \left(\frac{\log^{5d+5} n}{\eps^4} \right)$ \gpis queries with high probability\footnote{high probability means a probability of at least $1 - n^{-O(1)}$}, where the constant in $\Oh_d(\cdot)$ is a function of $d$.
 \end{theo}

Recently, it came to our notice that concurrently and independently, Dell et al.~\cite{DellLM19} obtained polylogarithmic query complexity for the hyperedge estimation problem using a similar query model. We will discuss their result shortly.

\subsection{Other related works}
\noindent Graph parameter estimation, where one wants to estimate the number of edges, triangles or small subgraphs in a graph, is a well-studied area of research in sub-linear algorithms. Feige \cite{Feige06}, and Goldreich and Ron \cite{GoldreichR08} gave algorithms to estimate the number of edges in a graph using degree, and degree and neighbour queries, respectively. Eden \etal~\cite{EdenLRS15} estimated the number of triangles in a graph using degree, neighbour and edge existence queries, and gave almost matching lower bound on the query complexity. This result was generalized for estimating the number of cliques of size $k$ in~\cite{EdenRS18}. Since the information revealed by degree, neighbour and edge existence queries is limited to the locality of the queried vertices, these queries are known as local queries \cite{G2017}. The subset queries, used in \cite{BeameHRRS18,Bhatta-abs-1808-00691, BhattaISAAC,DellLM19}, are examples of global queries, where a global query can reveal information of the graph at a much broader level.

\remove{
We study graph parameter estimation problems, for example estimating the number of edges or triangles in the graph in the query model of computation. In the query model we don't have direct access to the graph, instead we are allowed to query the graph. Graph parameter estimation using query access to the graph is a well studied problem. Feige \cite{Feige06}, and Goldreich and Ron \cite{GoldreichR08} gave algorithms to estimate the number of edges in a graph using degree, and degree and neighbour queries, respectively. Eden \etal~\cite{EdenLRS15} estimated the number of triangles in a graph using degree, neighbour and edge existence queries, and gave almost matching lower bound on the query complexity. Since the information revealed by degree, neighbour and edge existence queries is limited to the locality of the queried vertices, these queries are known as local queries \cite{G2017}. The subset queries, defined next, are examples of global queries, where a global query can reveal information of the graph at a much broader level.}

Goldreich and Ron \cite{GoldreichR08} solved the edge estimation problem in undirected graphs using $\tilde{O}(n/\sqrt{m})$ local queries. Dell and Lapinskas~\cite{DellL18} used the {\sc Independent set} ({\sc IS}) oracle to estimate the number of edges in bipartite graphs, where an {\sc IS} oracle takes a subset $S$ of the vertex set as input and outputs whether $S$ is an independent set or not. Their algorithm for edge estimation in bipartite graphs makes polylogarithmic {\sc IS} queries and $\Oh(n)$\remove{~\footnote{$n$ and $m$ denote the number of vertices and edges in the input graph.}} edge existence queries. Beame et al.~\cite{BeameHRRS18} extended the above result for the edge estimation problem in bipartite graphs to general graphs, and showed that the edge estimation problem in general graphs can be solved using $\tOh\left(\min \{\sqrt{m},n^2/m\}\right)$~\footnote{$\tOh(\cdot)$ hides a polylogarithmic term.} {\sc IS} queries. Recently, Chen et al.~\cite{CLW2019} improved this result to solve the edge estimation problem using only $\tOh\left(\min \{\sqrt{m},n/\sqrt{m}\}\right)$ {\sc IS} queries. \remove{Note that there is a large gap between the query complexities for the edge estimation problem in bipartite and general graphs, using {\sc IS} queries.}

\subsection{Setup and Notations}
\noindent We denote the sets $\{1,\ldots,n\}$ and $\{0,\ldots,n\}$ by $[n]$ and $[n^*]$, respectively. A hypergraph $\cH$ is a \emph{set system} $(U(\cH),\cF(\cH))$, where $U(\cH)$ denotes the set of vertices and $\cF(\cH)$ denotes the set of \remove{unordered} hyperedges. The set of vertices present in a hyperedge $F \in \cF(\cH)$ is denoted by $U(F)$ or simply $F$. A hypergraph $\cH$ is said to be $d$-uniform if all the hyperedges in $\cH$ consist of exactly $d$ vertices. The cardinality of the hyperedge set is $m(\cH)=\size{\cF(\cH)}$. For $u \in U(\cH)$, $\cF(u)$ denote the set of \remove{unordered} hyperedges that are incident on $u$. For $u \in U(\cH)$, the degree of $u$ in $\cH$, denoted as $\deg_{\cH}(u)=\size{\cF(u)}$ is the number of hyperedges incident on $u$. For a set $A$ and $a \in \N$, $A,\ldots,A$ ($a$ times) will be denoted as $A^{[a]}$. Let $\dsubset \subseteq U(\cH)$ be such that for every $i,j \in [d]$ either $A_i=A_j$ or $A_i \cap A_j =\emptyset$. 
This has a bearing on the \gpis{} oracle queries we make; either the sets we query with are disjoint, or are the same. Consider the following $d$-partite sub-hypergraph of $\cH$: $\left( U(\dsubset), \cF(\dsubset) \right)$ where the vertex set is $U(\dsubset) =\bigcup_{i=1}^d A_i$ and the hyperedge set is $\cF(\dsubset) = \left\{ \{ i_1, \ldots, i_d \}~|~ i_j \in A_j \right\}$; we will denote this $d$-partite sub-hypergraph of $\cH$ as $\cH(\dsubset)$. With this notation, $\cH\left( U^{[d]} \right)$ makes sense as a $d$-partite sub-hypergraph on a vertex set $U$.
The number of hyperedges in $\cH(\dsubset)$ is denoted by $m(\dsubset)$.
\paragraph{Ordered hyperedge} We will use the subscript $o$ to denote the set of ordered hyperedges. For example, $\cF_o(\cH)$ denotes the set of ordered hyperedges,  $m_o(\cH)$ denote $\size{\cF_o(\cH)}$, $\cF_o(u)$ denote the set of ordered hyperedges incident on $u$.  The ordered hyperedge set puts an order on the vertices such that $i$-th vertex of a hyperedge comes from $A_i$. Formally, $\cF_o(\dsubset)=\{F_o \in \cF_o(\cH):\mbox{the $i$-th vertex of $F_o$ is in $A_i, \forall i \in [d]$}\}$. The corresponding number for ordered hyperedges is $m_o(\dsubset)$. We have the following relation between $m(\dsubset)$ and $m_o(\dsubset)$. 
\begin{fact}
For $s \in [d]$, $m_o(A_1^{[a_1]},\ldots,A_s^{[a_s]})=m(A_1^{[a_1]},\ldots,A_s^{[a_s]})~ \times ~ \prod\limits_{i=1}^{s} a_i!$, where $a_i  
\in [d]$ and $\sum_{i=1}^s a_i = d$.
\label{obs:1}
\end{fact}

For a set $\cal P$, ``$\cal P$ is \colored with [n]'' means that elements of $\cal P$ is assigned a color out of $[n]$ colors independently and uniformly at random. Let $\E[X]$ and $\V[X]$ denote the expectation and variance of a random variable $X$. For an event $\cE$, $\overline{\cE}$ denotes the complement of $\cE$. Throughout the paper, the statement that ``event $\cE$ occurs with high probability'' is equivalent to $\pr(\cE) \geq 1-\frac{1}{n^c}$, where $c$ is an absolute constant. The statement ``$a$ is an $1 \pm \eps$ multiplicative approximation of $b$'' means $\size{b-a} \leq \eps \cdot b$. For $x \in \R$, $\exp(x)$ denotes the standard exponential function $e^x$. We denote $[k]\times \cdots \times [k]\mbox{~($p$ times)}$ using $[k]^{p}$, where $p \in \N$.\remove{ For a function $h:[k]^{d-1}\rightarrow [k]$, $\Lambda(h)=\{(c_1,\ldots,c_{d-1},c_d) ~:~c_1,\ldots,c_{d-1},c_d \in [k]~\mbox{and}~c_d=h(c_1,\ldots,c_{d-1})\}$. $\Lambda'(h)$ be the maximal subset of $\Lambda(h)$ such that no two tuples in $\Lambda'(h)$ are \emph{isomorphic}. $\Pi(h)=\{(c_{\pi(1)},\ldots,c_{\pi(d)})~:~(c_1,\ldots,c_d) \in \Lambda'(h)~\mbox{and}~\pi:[d]\rightarrow[d]~\mbox{is a bijection}\}$.} For us, $d$ is a constant.  $\Oh_d(\cdot)$ denotes the standard $\Oh(\cdot)$ where the constant depends on $d$. By polylogarithmic, we mean {$\Oh_d\left( \frac{(\log n)^{\Oh_{d}(1)}}{{\eps}^{\Oh(1)}}\right)$. The notation $\tOh_d(\cdot)$ hides a polylogarithmic term in $\Oh_d(\cdot)$.}

\remove{
As stated earlier, a pertinent question remains -- could we have estimated hyperedges using reasonable number of local queries? We establish the following lower bounds on \hest, thereby, establishing  the need for subset type 
of queries. 
Theorem~\ref{theo:main_local} is about the algorithm for \hest problem using degree and neighbor queries; and Theorem~\ref{theo:lower_local} gives the matching lower bound even if we use core and adjacency query along with degree and neighbor query. Both the theorems stated below are proved in Appendix~\ref{algo:edge_local}. 
\begin{theo}
\label{theo:main_local}
Let $\cH$ be a graph with $\size{U(\cH)}=n$. For any
$\eps >0$, \hest can be solved using $\widetilde{\cO}_d \left( \frac{n}{(m(\cH))^{1/d}} \right)$  queries with high probability, where a query is either a degree query or a neighbor query.
\end{theo}

\begin{theo}
\label{theo:lower_local}
Any multiplicative approximation algorithm that estimates the number of hyperedges in a hypergraph $\cH$ with constant probability, requires $\Omega_d\left( \frac{n}{(m(\cH))^{1/d}}\right)$ queries, where allowed queries are degree, neighbor, adjacency and core query. 
\end{theo}
}

\subsection{Paper organization}
\label{ssec:paperorg}
\noindent
We define in Section~\ref{sec:prelim_actual} two other query oracles, $\mbox{{\sc GPIS}}_1$ and $\mbox{{\sc GPIS}}_2$ that can be simulated by using polylogarithmic \gpis queries. The role of these two oracles is mostly expository -- they help us to describe our algorithms in a neater way. Section~\ref{sec:overview} gives a broad overview of our query algorithm that involves exact estimation, sparsification, coarse estimation and sampling. Section~\ref{sec:compare} contextualizes our work vis-a-vis recent works~\cite{BeameHRRS18,Bhatta-abs-1808-00691,BhattaISAAC,DellLM19}.
The novel contribution of this work is sparsification which is given in detail in Section~\ref{sec:sparse}. Sections~\ref{sec:exact} and \ref{sec:coarse} consider the proofs for exact and coarse estimation, respectively. The algorithm and its proof of correctness are discussed in Sections~\ref{sec:algo} and \ref{sec:correct}, respectively. The equivalence proofs of the \gpis{} oracle and its variants, that were introduced in Section~\ref{sec:prelim_actual}, are discussed in Appendix~\ref{append:gpis}.  Some useful probability results are given in Appendix~\ref{sec:prelim}. As this paper talks of different kinds of oracles, we have put all the definitions at one place in Appendix~\ref{sec:oracle-def} for ease of reference.

\section{Preliminaries: \gpis oracle and its variants} \label{sec:prelim_actual}

\noindent Note that the \gpis query oracle takes as input $d$ pairwise disjoint subsets of vertices. We now define two related query oracles \gpis{}$ _1$ and \gpis{}$_2$ that remove the disjointness criteria on the input. We show that both these query oracles can be simulated by making polylogarithmic queries to the \gpis oracle with high probability. \gpis{}$_1$ and \gpis{}$_2$ oracles will be used in the description of the algorithm for ease of exposition.

\begin{description}
	\item[($\mbox{{\sc GPIS}}_1$)] Given $s$ pairwise disjoint subsets of vertices $A_1,\ldots,A_s \subseteq U(\cH)$ of a hypergraph $\cH$ and $a_1,\ldots,a_s \in [d]$ such that $\sum_{i=1}^{s}a_i =d$, \gonepis query oracle on input $A_1^{[a_1]},A_2^{[a_2]},\cdots,A_s^{[a_s]}$ answers {\sc Yes} if and only if $m(\ssubset) \neq 0$.

	\item[($\mbox{{\sc GPIS}}_2$)] Given any $d$ subsets of vertices $\dsubset \subseteq U(\cH)$ of a hypergraph $\cH$, \gtwopis query oracle on input $A_1,\ldots,A_d$ answers  {\sc Yes}  if and only if $m(\dsubset) \neq 0$.
\end{description}

Observe that the \gtwopis query oracle is the same as the \gpis query oracle without the requirement that the input sets are disjoint. For the \gonepis query oracle, multiple repetitions of the same set is allowed in the input. It is obvious that a \gpis query can be simulated by a \gonepis or \gtwopis query oracle. Using the following observations, whose proofs are given in Appendix~\ref{append:gpis}, we show how a \gonepis or a \gtwopis query can be simulated by polylogarithmic \gpis queries. 

\begin{obs} \label{obs:queryoracles}
\begin{itemize}
	\item[(i)] A \gonepis query can be simulated using polylogarithmic \gpis queries with high probability.
	\item[(ii)] A \gtwopis query can be simulated using $2^{\Oh(d^2)}$ \gonepis queries.
	\item[(iii)] A \gtwopis query can be simulated using polylogarithmic \gpis queries with high probability.
\end{itemize}
\end{obs}
To prove Theorem~\ref{theo:main_restate}, we first consider the following lemma. This lemma is the central result of the paper and from it, the main theorem (Theorem~\ref{theo:main_restate}) follows. 
\begin{lem}
\label{lem:prob1}
Let $\cH=(U(\cH),\cF(\cH))$ be a hypergraph with $n$ vertices, i.e., $\size{U(\cH)}=n$. For any~$\eps > \lbeps$, \hest can be solved with probability {at least} $1-\frac{1}{n^{4d}}$ and using $\Oh\left( \eps^{-4} \log^{5d+4} n \right)$ queries, where each query is either a \gonepis query or a \gtwopis query.
\end{lem}
Assuming Lemma~\ref{lem:prob1} to be true, we now prove Theorem~\ref{theo:main_restate}.
\begin{proof}[Proof of Theorem~\ref{theo:main_restate} ]
If $\eps \leq \lbeps$, we make a \gpis{} query with $(\{a_1\},\ldots,\{a_d\})$ for all distinct $a_1,\ldots,$ $a_d \in U(\cH)=U$ and enumerate by brute force the exact value of $m_o(\cH)$. So, we make at most $n^d=\Oh_d\left(\eps^{-4} \log^{5d+5} n \right)$ many \gpis queries as $\eps \leq \lbeps$.
 If $\eps > \lbeps$, we use the algorithm corresponding to Lemma~\ref{lem:prob1}, where each query is either a \gonepis query or a \gtwopis query. However, by Observation~\ref{obs:queryoracles}, each \gonepis and \gtwopis query can be simulated by $\Oh_d(\log n)$ many \gpis queries with high probability. So, we can replace each step of the algorithm, where we make either \gonepis or \gtwopis query, by $\Oh_d(\log n)$ many \gpis queries. Hence, we are done with the proof of Theorem~\ref{theo:main_restate}.
\end{proof}
In the rest of the paper, we mainly focus on proving Lemma~\ref{lem:prob1}.

\section{Technical Overview} \label{sec:overview}
\noindent
We briefly describe the overview of our work and put our work in context with recent works. 
\subsection{The context of our work}
\noindent Beame \etal~\cite{BeameHRRS18} developed a framework to estimate the number of edges in a graph using \bis queries. This framework involves subroutines for sparsifying a graph into a number of subgraphs each with reduced number of edges, and exactly or approximately counting the number of edges in these subgraphs. Sparsification constitutes the main building block of this framework. The sparsification routine of Beame \etal~randomly colors the vertices of a graph with $2k$ colors. Let $A_i$ denote the set of vertices colored with color $i$. Beame \etal~argued that the sum of the number of edges between $A_i$ and $A_{k+i}$ for $1\leq i\leq k$ approximates the total number of edges in the graph within a constant factor with high probability. Therefore, the original problem reduces to the problem of counting the number of edges in bipartite subgraphs. Here, we have $k$ bipartite subgraphs and it suffices to count the number of edges in these bipartite subgraphs. With this method, they obtained a query complexity of $O(\eps^{-4} \log^{14} n)$ with high probability.

Bhattacharya \etal~\cite{BhattaISAAC} extended this framework to estimate the number of triangles in a graph using \tis queries, where given three disjoint subsets of vertices $A, \, B, \, C$, a \tis query with inputs $A, \, B, \, C$ answers whether there exists a triangle with one endpoint in each of $A,B,C$. The sparsification routine now requires to color the vertices randomly using $3k$ colors and counts the number of properly colored triangles where a properly colored triangle has one endpoint each in $A_i$, $A_{k+i}$, and $A_{2k+i}$ for any $1\leq i\leq k$. Unlike the scenario in Beame \etal~where two edges intersect in at most one vertex, here two triangles can share an edge. Therefore, the random variables used to estimate the number of {\it properly colored} triangles are not independent. Bhattacharya \etal~estimated the number of triangles assuming that the number of triangles incident on any edge is bounded by a parameter $\Delta$. In this way, they obtained a query complexity of $O(\eps^{-4} \Delta^{2} \log^{18}n)$.

In this paper, we fully generalize the frameworks of Stockmeyer~\cite{Stockmeyer83,Stockmeyer85}, Ron and Tsur~\cite{RonT16}, Beame \etal~\cite{BeameHRRS18} and Bhattacharya~\etal~\cite{BhattaISAAC} to estimate the number of hyperedges in a $d$-uniform hypergraph using \gpis queries. \remove{We do not have the dependency problem for estimating the number of hyperedges in a hypergraph because existence of one hyperedge does not imply anything about the existence of another hyperedge with a subset of its vertices.}

\remove{\comments{In subset size estimation problem using subset query oracle, we want to estimate the size of a subset $S\subseteq T$ using an oracle which tells us, given a query set $Q\subseteq T$, whether $Q \cap T$ is empty or not. Viewed differently, this is about estimating an unknown set $S$ by looking at its intersection pattern with a known set $T$.}

The hyperedge estimation problem can be thought of as a generalized version of the subset size estimation problem. We show using an induction on $d$ that \hest using \gpis queries in a $d$-uniform hypergraph generalizes from the subset size estimation problem.}

\subsection{Our work in a nutshell}
\begin{figure}[h!]
	\centering
	\includegraphics[width=\linewidth]{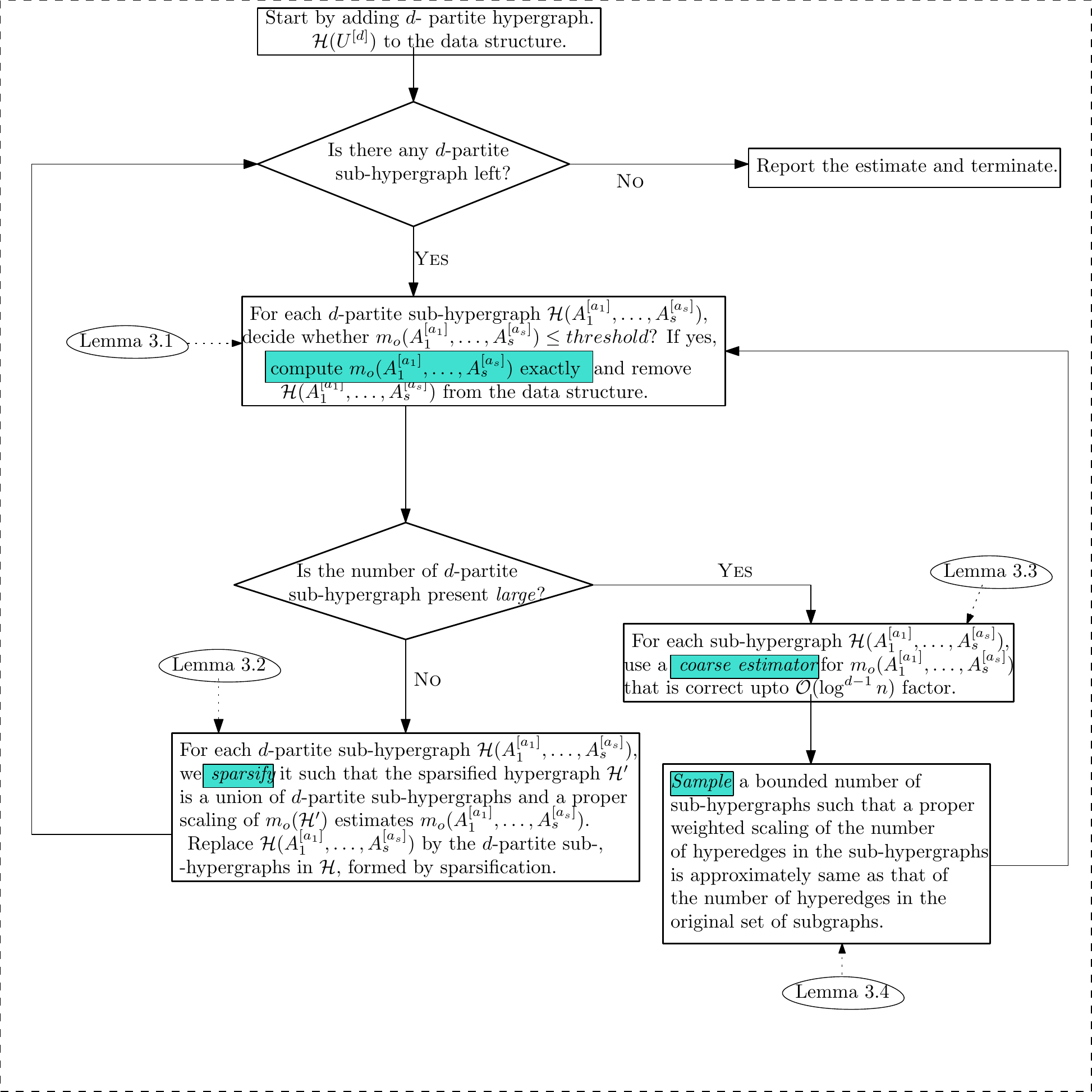}
	\caption{Flow chart of the algorithm. The highlighted texts indicate the basic building blocks of the algorithm. We also indicate the corresponding lemmas that support the building blocks. \remove{Gopi: Correct the lemma numbers later.}}
	\label{fig:flowchart}
\end{figure}

\paragraph{The algorithmic framework} In Figure~\ref{fig:flowchart}, we give a flowchart of the algorithm. Our algorithm begins by adding the $d$-partite hypergraph $\cH(U^{[d]})$ in a suitable data structure, let us call it $\cD$. For each hypergraph in $\cD$, using the \emph{exact estimation} process, we \emph{exactly count} the number of hyperedges in a hypergraph if the number of hyperedges is less than a \emph{threshold} $\tau$ and remove the corresponding hypergraph from $\cal D$. For hypergraphs in which the number of hyperedges is more than $\tau$, we can not do \emph{exact estimation}. So, we resort to \emph{sparsification} if the number of  hypergraphs left in $\cD$ is not \emph{large}, i.e., below a \emph{threshold} $\cN$. Hypergraph sparsification breaks a $d$-partite hypergraph into a disjoint union of $d$-partite sub-hypergraphs. As sparsification goes on in iterations, it may populate the data structure $\cD$ with \emph{significantly many} hypergraphs. In that case, we do a \emph{coarse estimation} followed by \emph{sampling} to have a reasonable number of sub-hypergraphs and then go back to the exact estimation step and continue in a loop. \gonepis and \gtwopis queries will be used for exact and coarse estimations, respectively; their logarithmic equivalence with \gpis will prove the final result. For ease of analysis, we consider {\it ordered} hyperedges. Fact \ref{obs:1} relates $m(A_1,\ldots,A_d)$ and $m_o(A_1,\ldots,A_d)$, the number of unordered and ordered hyperedges, respectively. 
\remove{
Since \gonepis and \gtwopis queries can be simulated using at most logarithmic \gpis queries, in the description of our algorithm, we use \gonepis and \gtwopis queries interchangeably, the usage of the particular query would be clear from the context. For ease of analysis we consider {\it ordered} hyperedges. The relation between the number of unordered hyperedges $m(A_1,\ldots,A_d)$ and the number of ordered hyperedges $m_o(A_1,\ldots,A_d)$ in sets $A_1,\ldots,A_d$ is given as Fact \ref{obs:1}.
}

\remove{
\paragraph{Exact estimation.} In this step, we look at a $d$-partite sub-hypergraph and decide whether $m_o(U^{[d]})$, the number of ordered hyperedges, is larger or smaller than a threshold $\tau$~\footnote{Threshold $\tau$ will be fixed later in Section~\ref{sec:exact}.}, and if $m_o(U^{[d]}) \leq \tau$, compute the exact value of $m_o(U^{[d]})$ using Lemma~\ref{lem:exact}, whose proof is in Section~\ref{sec:exact}. If $m_o(U^{[d]})$ is greater than $\tau$ \complain{and the number of hypergraphs in $\cD$ is below a threshold}, we move to the sparsification step, else to the coarse estimation step.  \comments{Gopi:Can we change this paragraph a bit as follows?}
}

\paragraph{Exact estimation.} In this step, we look at a $d$-partite sub-hypergraph $\cH(\ssubset)$ and decide whether $m_o(\ssubset)$, the number of ordered hyperedges, is larger or smaller than a threshold $\tau$~\footnote{Threshold $\tau$ will be fixed later in Section~\ref{sec:exact}.}, and if $m_o(\ssubset) \leq \tau$, compute the exact value of $m_o(\ssubset)$ using Lemma~\ref{lem:exact}, whose proof is in Section~\ref{sec:exact}, and then delete it from $\cD$. If the number of hypergraphs in $\cD$ is below the threshold $\cN$, with the number of ordered hyperedges in it more than $\tau$, we move to the sparsification step, else to the coarse estimation step.

\begin{lem} [Exact Estimation] \label{lem:exact} There exists a deterministic algorithm $\cAe$ that takes as input -- a $d$-uniform hypergraph $\cH$, constants $a_1,\ldots,a_s \in [d]$ such that $\sum_{i=1}^s a_i =d$ where $s \in [d]$, pairwise disjoint subsets $A_1,\ldots,A_s$ of $U(\cH)$, and a threshold parameter $\tau \in \N$ -- and decides whether $m_o(\ssubset) \leq \tau$ using $\cO_d(\tau \log n)$ \gonepis queries. Moreover, $\cAe$ finds the exact value of $m_o(\ssubset)$ when $m_o(\ssubset) \leq \tau$. \end{lem}

\paragraph{Sparsification.} 
We {\sc color} $U(\cH)$, the vertices of hypergraph, with $[k]$ colors to sparsify the $d$-partite hypergraph $\cH(U^{[d]})$ so that 
\begin{itemize}
	\item[(i)] the sparsified hypergraph consists of a set of $d$-partite sub-hypergraphs and  
	\item[(ii)] a proper scaling of the sum of the number of ordered hyperedges in the sub-hypergraphs is a \emph{good} estimate of $m_o(U^{[d]}),U=U(\cH),$ with high probability. 
\end{itemize}
The sparsification result is formally stated next. The proof uses the \emph{method of averaged bounded differences} and \emph{Chernoff-Hoeffding} inequality. The detailed proof is given in Section~\ref{sec:sparse}. The heart of our work is the general sparsification routine that we believe will find independent uses. 
\begin{lem}[Sparsification] \label{lem:sparse}
Let $\cH$ be any $d$-uniform hypergraph and $k\geq 1$ be any positive integer. Let
        \begin{itemize}
                \item $h_d:[k]^d\rightarrow \{0,1\}$ be a hash function such that independently for any $d$-tuple ${\bf a}\in [k]^d$, $\pr[h_d({\bf a})=1]=\frac{1}{k}$.
                \item $A_1,\ldots,A_s$ be any pairwise disjoint subsets of the vertex set $U(\cH)$, where $1\leq s\leq d$. Let us choose any $a_i\in [d]$ vertices from $A_i$ such that $\sum_{i=1}^s a_i=d$. 
                \item Vertices in $A=\bigcup_{i=1}^s A_i$ are \colored with $[k]$. $\chi(i,j)$'s denote the color classes for each $A_i$, that is, $\chi(i,j)=\{v\in A_i:~v~\mbox{is colored with color $j$}\}$, where $i \in [s]$ and $j \in [k]$. 
                \item A hyperedge $(x_1,\ldots,x_d)$ is said to be \emph{properly colored} if $h(c_1,\ldots,c_d)=1$, where $c_i$ is the color of $x_i$. Let $\cR_d$ denote the number of \emph{properly colored} hyperedges defined as follows. $$\cR_d=\sum\limits_{(c_1,\ldots,c_d) \in [k]^d} m_o\left(\chi({1,c_1}),\ldots,\chi({1,c_{a_1}}),\ldots \ldots, \chi({s,c_{d-a_s+1}}),\ldots,\chi({s,c_{d}})\right)\times h_d(c_1,\ldots,c_d).$$
        \end{itemize}
Then, for a suitable constant $\theta > d$ and $p_d = \frac{d!}{n^{4\theta -2d}}$,
$$ \pr \left(\size{ \cR_d  - \frac{m_o(A_1^{[a_1]},\ldots,A_s^{[a_s]})}{k}} \geq 2^{2d}\theta^d   \sqrt{ d! ~ m_o(A_1^{[a_1]},\ldots,A_s^{[a_s]}) \log^d n} \right) \leq p_d.$$
\end{lem}
Sparsification ensures that $m_o(U^{[d]})$, the number of ordered hyperedges between $A_1^{[a_1]},\ldots,A_s^{[a_s]}$, is approximately preserved when $m_o(U^{[d]})$ is above a threshold $\tau$.  

Assume that $m_o(U^{[d]})$ is \emph{large}~\footnote{{A fixed polylogarithmic quantity to be decided later in Section~\ref{sec:algo}.}} and $\cH(U^{[d]})$ has been sparsified. We add to the data structure $\cD$ a set of $d$-partite sub-hypergraphs obtained from the sparsification step. Refer to the flowchart in Figure~\ref{fig:flowchart} again. With the new $d$-partite sub-hypergraphs, we loop back to the exact estimation step. 
\remove{
For each $d$-partite hypergraph $\cH(\ssubset)$ in the data structure, we check whether $m_o(\ssubset)$ is less than the threshold $\tau$ using the above algorithm for exact estimation (Lemma~\ref{lem:exact}). If it is less than threshold $\tau$, we compute the exact value of $m_o(\ssubset)$ using Lemma~\ref{lem:exact} and remove $\cH(\ssubset)$ from the data structure. 
}
After the exact estimation step, we are left with some $d$-partite hypergraphs such that the number of ordered hyperedges in each hypergraph is more than the threshold $\tau$. If the number of such hypergraphs {present in $\cD$} is not large, i.e., below $\cN$, then we sparsify each hypergraph $\cH(\ssubset)$ using the algorithm corresponding to Lemma~\ref{lem:sparse}.  
\remove{ such that
\begin{itemize}
	\item[(i)] we obtain some $d$-partite sub-hypergraphs for each of the sparsified hypergraphs, and 
	\item[(ii)] a constant scaling of the sum of the number of ordered hyperedges in the sparsified sub-hypergraphs approximates $m_o(\ssubset)$.
\end{itemize}
Gopi: This (i) and (ii) was written at the beginning also. Do we need it again? \comments{Yes, we can remove the above red marked text.}
}

\paragraph{Coarse estimation and sampling.} If we have a large number of $d$-partite sub-hypergraphs of $\cH(U^{[d]})$ and each sub-hypergraph contains a large number of ordered hyperedges, then we \emph{coarsely} estimate the number of ordered hyperedges in each sub-hypergraph; see Figure~\ref{fig:flowchart}. Our estimator is correct up to a $\Oh_d(\log^{d-1} n)$ factor using the algorithm corresponding to the following lemma, whose proof is given in Section~\ref{sec:coarse}. 
\begin{lem}[Coarse Estimation] \label{lem:coarse_main} There exists an algorithm $\cAc$ that takes as input $d$ subsets $A_1,\ldots, A_d$ of vertex set $U(\cH)$ of a $d$-uniform hypergraph $\cH$ and returns $\hat{E}$ as an estimate for $m_o(\dsubset)$ such that 
	$$ \frac{m_o(\dsubset)}{8d^{d-1}2^d \log^{d-1} n} \leq \hat{E} \leq 20d^{d-1}2^d \cdot m_o(\dsubset) \log^{d-1} n $$
with probability at least $1-n^{-8d}$. Moreover, the number of \gtwopis queries made by the algorithm is $\Oh_d(\log ^{d+1} n)$. \end{lem} 

First, we coarsely estimate the number of ordered hyperedges in each sub-hypergraph. Then we sample a set of sub-hypergraphs such that a weighted sum of the number of hyperedges in the sample approximately preserves the sum of the number of ordered hyperedges in the sub-hypergraphs. This kind of sampling technique, also known as the importance sampling, is given for the edge estimation problem by Beame~\etal~\cite{BeameHRRS18}. The lemma corresponding to this sampling technique is formally stated as Lemma~\ref{lem:importance1} in Appendix~\ref{sec:prelim}. The importance sampling lemma that we require is stated as follows. 

\begin{lem}[Importance Sampling] \label{lem:importance-app} 
Let $\{(A_{i1},\ldots,A_{id},w_i):i \in [r]\}$ be the set of tuples in the data structure $\cD$ and $e_i$ be the coarse estimate for $m_o(A_{i1},\ldots,A_{id},w_i)$ such that (i) $w_i,e_i \geq 1~ \forall i \in [r]$, (ii) $\frac{e_i}{\alpha} \leq m_o(A_{i1},\ldots,A_{id}) \leq e_i \cdot \alpha$ for some $\alpha >0$ and $\forall ~i \in [r]$, and (iii) $\sum_{i=1}^r {w_i\cdot m_o(A_{i1},\ldots,A_{id})} \leq M$. Then, there exists an algorithm that finds a set $\{(A'_{i1},\ldots,A'_{id},w'_i):i \in [r']\}$ of tuples, with probability at least $1-\delta$, such that the above three conditions hold and 
\begin{equation*}
\size{\sum_{i=1}^{r'} {w'_i\cdot m_o(A'_{i1},\ldots,A'_{id})}-\sum_{i=1}^r {w_i\cdot m_o(A_{i1},\ldots,A_{id})}} \leq \lambda S \, , 
\end{equation*}
where $S=\sum_{i=1}^r {w_i m_o(A_{i1},\ldots,A_{id})}$, $\lambda>0$, $\delta >0$ and $r'=\cO\left(\frac {\alpha^4 \log M}{\lambda^{2}} \left(\log \log M + \log \frac{1}{\delta}\right)\right)$. 
\end{lem}

\paragraph{Putting things together.} The data structure $\cal D$ that our iterative algorithm uses (see the flowchart in Figure~\ref{fig:flowchart}) has an accumulator $\Psi$ for the number of hyperedges and a set of tuples representing the $d$-partite sub-hypergraphs along with their weights -- these $d$-partite sub-hypergraphs are generated from sparsification. The algorithm uses the \emph{exact estimation} step (viz. Algorithm~\ref{algo:exact}) to figure out if the number of hyperedges $m_o(\dsubset) \leq \tau$ using Lemma~\ref{lem:exact}. If $m_o(\dsubset) \leq \tau$, we add $w \cdot m_o(\dsubset)$ to $\Psi$ and remove the tuple $(A_1,\ldots,A_d,w)$ from $\cal D$. If the number of tuples in $\cal D$ is at most $\cN = \kappa_d \cdot \frac{\log ^{4d} n}{\eps^2}$, then we carry out a sparsification step, else we do an importance sampling; a coarse estimation step (viz. Algorithms~\ref{algo:verify} and~\ref{algo:coarse}) acts as a pre-processing step for importance sampling. We do either sparsification or coarse estimation followed by importance sampling in an iteration because their roles are complementary. While successive sparsifications ease estimating the number of hyperedges by breaking a $d$-partite hypergraph into a disjoint union of $d$-partite sub-hypergraphs, the number of such $d$-partite sub-hypergraphs may be as high as $4^d \cdot \cN =4^d \kappa_d \cdot \frac{\log ^{4d} n}{\eps^2}$ (see Observation~\ref{obs:num_tuples}). Coarse estimation followed by importance sampling (see Lemma~\ref{lem:importance-app}) ensures that the effect of sparsification is neutralized so that the number of tuples in $\cal D$ is at most $\cN$. This ensures that out of two successive iterations of the algorithm, there is at least one sparsification step. The coarse estimation applies when the number of tuples $r>\cN=\frac{\kappa_d\log^{4d} n}{\eps^2}$. 
\remove{For each tuple in $\cal D$, we find an estimate $\hat{E_i}$ such that $\frac{m_o(A_1,\ldots,A_d)}{\alpha} \leq \hat{E_i} \leq \alpha m_o(A_1,\ldots,A_d)$ using $\Oh_d(\log^{d+1} n)$ \gtwopis queries per tuple (see Lemma~\ref{lem:coarse_main} for the precise value of $\alpha$).}
For each tuple in $\cal D$, we find an estimate $\hat{E_i}$ using $\Oh_d(\log^{d+1} n)$ \gtwopis queries per tuple (see Lemma~\ref{lem:coarse_main}).
This estimate serves as condition (ii) for the application of Lemma~\ref{lem:importance-app}. 
\remove{For the $r$ tuples in $\cal D$, let $e_i$ be the coarse estimate of $m_o(A_{i1},\ldots,A_{id},w_i)$, $i \in [r]$, such that (i) $w_i,e_i \geq 1~ \forall i \in [r]$, (ii) $\frac{e_i}{\alpha} \leq m_o(A_{i1},\ldots,A_{id}) \leq e_i \cdot \alpha$ for some $\alpha >0$ and $\forall ~i \in [r]$, and (iii) $\sum_{i=1}^r {w_i\cdot m_o(A_{i1},\ldots,A_{id})} \leq M$.} Because of the importance sampling (see Lemma~\ref{lem:importance-app}), one can replace the $r$ tuples in $\cal D$ with a sample of $r'$ tuples such that the required estimate is approximately maintained. Since $r'\leq \kappa_d\cdot \frac{\log^ {4d} n}{\eps^2}=\cN$, the next step will be sparsification.

Observe that the query complexity of each iteration is polylogarithmic. Note that the number of ordered hyperedges reduces by a constant factor after each sparsification step. So, the number of iterations is bounded by $\Oh_d(\log n)$. Hence, the query complexity of our algorithm is polylogarithmic. This completes a high level description of our algorithm. 

\remove{
Using the sampling scheme mentioned above, we check whether $m_o(\ssubset)$ is less than threshold $\tau$ using Lemma~\ref{lem:exact} for each $d$-partite hypergraph $\cH(\ssubset)$. If $m_o(\ssubset)<\tau$, we compute the exact value of $m_o(\ssubset)$ using Lemma~\ref{lem:exact} and remove $\cH(\ssubset)$ from the data structure. Otherwise, we perform sparsification, exact counting and coarse estimation to estimate $m_o(\ssubset)$. Observe that the query complexity of each iteration is polylogarithmic. Note that the number of ordered hyperedges reduces by a constant factor after each sparsification step. So, the number of iterations is bounded by $\Oh_d(\log n)$. Hence, the query complexity of our algorithm is polylogarithmic. This completes a high level description of our algorithm. 
}

\subsection{Our work vis-a-vis some recent works} \label{sec:compare}

\paragraph*{Comparison with Beame et al.~\cite{BeameHRRS18} and Bhattacharya et al.~\cite{Bhatta-abs-1808-00691,BhattaISAAC}:}
Beame et al.~\cite{BeameHRRS18} showed that {\sc Edge Estimation} problem can be solved using $\Oh(\eps^{-4}\log ^{14} n)$ many \bis queries. Note that a \bis query is a special case of \gpis query for $d=2$. Bhattacharya et al.~\cite{Bhatta-abs-1808-00691,BhattaISAAC} first considered the {\sc Triangle Estimation} problem in graphs using a \tis query oracle, a generalization of \bis query oracle. \remove{As a first step towards the generalization of {\sc Edge Estimation} using \bis towards \hest using \gpis, Bhattacharya et al.~\cite{Bhatta-abs-1808-00691,BhattaISAAC} considered {\sc Triangle Estimation} in graphs using a query oracle called \tis.} They showed that an $(1 \pm \eps)$-approximation to the number of triangles can be found by using $\Oh(\Delta^{2} \eps^{-4} \log ^{18} n)$ many \tis queries, where $\Delta$ denotes the maximum number of triangles sharing an edge. Note that the {\sc Triangle Estimation} problem in graphs using \tis queries is analogous to the \hest problem in $3$-uniform hypergraphs using a special \gpis query when $d=3$. In this paper, we build on the works of Beame et al.~\cite{BeameHRRS18} and Bhattacharya et al.~\cite{Bhatta-abs-1808-00691,BhattaISAAC} for \hest using \gpis queries. Our main contribution in this work is twofold. First, our work can be seen as a complete generalization of the works of Beame et al.\cite{BeameHRRS18} and Bhattacharya et al.~\cite{Bhatta-abs-1808-00691,BhattaISAAC} to $d$-uniform hypergraphs. Secondly, our query complexity bounds are independent of $\Delta$, this dependence on $\Delta$ was present in the work of Bhattacharya et al.~\cite{Bhatta-abs-1808-00691,BhattaISAAC}. \remove{Our results not only remove the dependency on $\Delta$ present in the work of Bhattacharya et al.~\cite{Bhatta-abs-1808-00691,BhattaISAAC}, but also generalizes the results of Beame et al.\cite{BeameHRRS18} and Bhattacharya et al.~\cite{Bhatta-abs-1808-00691,BhattaISAAC} to $d$-uniform hypergraph setting} As mentioned earlier, the extension from the {\sc Edge Estimation} problem using \bis queries to the {\sc Hyperedge Estimation} problem using \gpis queries is not obvious as edges in a graph can intersect in at most one vertex, but the intersection pattern of hyperedges in hypergraphs is complicated. The sparsification algorithm in the work of Beame et al.~\cite{BeameHRRS18} uses heavily the fact that two edges in a graph can intersect with at most one vertex. This may be the reason why Bhattacharya et al.~\cite{Bhatta-abs-1808-00691,BhattaISAAC} required a bound on the number of triangles sharing an edge. In this paper, we show that an ingenious way of coloring the hyperedges in a $d$-uniform hypergraph and an involved use of induction on $d$ generalize the results of both Beame et al.~\cite{BeameHRRS18} and Bhattacharya et al.~\cite{Bhatta-abs-1808-00691,BhattaISAAC}. The main technical and non-trivial contribution in this paper is the sparsification algorithm, which enables us to generalize the works of Beame et al.~\cite{BeameHRRS18} and Bhattacharya et al.~\cite{Bhatta-abs-1808-00691,BhattaISAAC}.

\paragraph*{Comparison with Dell et al.~\cite{DellLM19}:} Concurrently and independently, Dell et al.~\cite{DellLM19} obtained similar results for the hyperedge estimation problem. They showed that one can find an $(1 \pm \eps)$-approximation to the number of hyperedges in a $d$-uniform hypergraph by using $\Oh_d\left(\frac{\log ^{4d+7} n}{\eps^2} \log \frac{1}{\delta}\right)$ \emph{colorful independent set} queries with probability at least $1-\delta$. A colorful independent set query is same as that of a \gpis query oracle considered in this paper; \gpis{} query was introduced by Bishnu et al.~\cite{BishnuGKM018,Bishnu18}. Note that the algorithm of Dell et al.~\cite{DellLM19} for \hest uses comparable but fewer number of \gpis queries. However, our algorithm is different and conceptually much simpler than that of Dell et al.~\cite{DellLM19}. Most importantly, the main technical tool, our sparsification step is completely different from theirs.

\section{Sparsification: Proof of Lemma~\ref{lem:sparse}} \label{sec:sparse}
\noindent
We start by reproducing the statement of the Lemma~\ref{lem:sparse} below. Sparsification ensures that the number of ordered hyperedges $m_o(A_1^{[a_1]},\ldots,A_s^{[a_s]})$ is approximately preserved across any pairwise disjoint subsets of vertices $A_1,\ldots,A_s$ and integers $a_1,\ldots,a_s$ such that $\sum_{i=1}^s a_i=d$.

\begin{lem}[Sparsification] \label{lem:sparse2}
Let $\cH$ be any $d$-uniform hypergraph and $k\geq 1$ be any positive integer. Let
	\begin{itemize}
		\item $h_d:[k]^d\rightarrow \{0,1\}$ be a hash function such that independently for any tuple ${\bf a}\in [k]^d$, $$\pr[h_d({\bf a})=1]=\frac{1}{k}$$.
		
		\item $A_1,\ldots, A_s$ be any pairwise disjoint subsets of $U(\cH)$, where $1\leq s\leq d$. Let us choose any $a_i \in [d]$ vertices from $A_i$ such that $\sum_{i=1}^s a_i=d$. \remove{Sparsification shows that the number of ordered hyperedges between $A_1^{[a_1]},\ldots,A_s^{[a_s]}$ is approximately preserved.
		Gopi: Do we need this in the lemma statement? If we change, we need to change in the earlier Section as well.}
		\item Vertices in $A=\bigcup_{i=1}^s A_i$ are \colored with $[k]$. $\chi(i,j)$'s denote the color classes for each $A_i$, that is, $\chi(i,j)=\{v\in A_i:~v~\mbox{is colored with color $j$}\}$, where $i \in [s]$ and $j \in [k]$. \remove{Vertices in $A=\cup A_i$ be \colored with $[k]$, where the color classes in $A_i$ is denoted as $\chi(i,j)=\{v\in A_i:~v~\mbox{is colored with color $j$}\}$, where $i \in [s]$ and $j \in [k]$.}
		\item A hyperedge $(x_1,\ldots,x_d)$ is said to be \emph{properly colored} if $h(c_1,\ldots,c_d)=1$, where $c_i$ is the color of $x_i$. Let $\cR_d$ denote the number of \emph{properly colored} hyperedges defined as follows. $$\cR_d=\sum\limits_{(c_1,\ldots,c_d) \in [k]^d} m_o\left(\chi({1,c_1}),\ldots,\chi({1,c_{a_1}}),\ldots \ldots, \chi({s,c_{d-a_s+1}}),\ldots,\chi({s,c_{d}})\right)\times h_d(c_1,\ldots,c_d).$$
		\remove{
		\item A hyperedge is said to be \emph{properly colored} if its {\sc COLOR} induces the hash function to evaluate to $1$. Let $\cR_d$ denote the number of \emph{properly colored} hyperedges defined as follows. $$\cR_d=\sum\limits_{(c_1,\ldots,c_d) \in [k]^d} m_o\left(\chi({1,c_1}),\ldots,\chi({1,c_{a_1}}),\ldots \ldots, \chi({s,c_{d-a_s+1}}),\ldots,\chi({s,c_{d}})\right)\times h_d(c_1,\ldots,c_d).$$
		}
	\end{itemize}   
Then, for a suitable constant $\theta > d$ and $p_d = \frac{d!}{n^{4\theta -2d}}$,
$$ \pr \left(\size{ \cR_d  - \frac{m_o(A_1^{[a_1]},\ldots,A_s^{[a_s]})}{k}} \geq 2^{2d}\theta^d   \sqrt{ d! m_o(A_1^{[a_1]},\ldots,A_s^{[a_s]}) \log^d n} \right) \leq p_d.$$
\end{lem}

Before getting into the formal proof of the lemma, let us explore the role of the hash function in sparsification. 

\subsection{{The role of the hash function in sparsification}}
\noindent 
We feel our work could generalize the framework of Beame et al.~\cite{BeameHRRS18} to hyperedge estimation mainly because of the new sparsification result aided by an involved use of induction. The sparsification result depends on the particular choice of hash function. As we move from edge to triangle and then to hyperedge in a $d$-uniform hypergraph, implicit structures blow up, and managing their interrelations also become difficult. Consider the following. Two edges can intersect in at most one vertex. two triangles can intersect in at most two vertices and two hyperedges in $d$-uniform hypergraph intersect in at most $d-1$ vertices. \remove{The existence of three triangles $(a,b,x)$, $(b,c,y)$ and $(a,c,z)$ in $G$ imply the existence of the triangle $(a,b,c)$ in $G$. The number of such implicit structures will blow up for hyperedges.}

To count the number of edges in $G$ using sparsification, Beame et al.~\cite{BeameHRRS18} \colored the vertices with $k$ colors and looked at the edges between certain pair of color classes. So, it boils down to counting the number of properly colored edges. This random coloring process can be encoded by the hash function $h:[k]\rightarrow [k]$. In our first effort in generalizing this work to count the number of triangles in a graph, we extended the above hash function to the following one -- $h:[k] \times [k] \rightarrow [k]$ in ~\cite{Bhatta-abs-1808-00691,BhattaISAAC}. Look at any edge $(a,b)$ in this graph under the above hash function. The edge is properly colored if the colors on the vertices of the edge come from specified color classes. We count all triangles on these properly colored edges, and scale it appropriately. Once an edge is not properly colored, we may miss it and hence, all triangles incident on that edge. This is why we need a bound on the number of triangles incident on an edge in~\cite{Bhatta-abs-1808-00691,BhattaISAAC}, otherwise the variance blows up and the concentration inequality gives poor result.

In this work, we consider a different hashing scheme with the hash function $h:[k]^d \rightarrow \{0,1\}$, in the $d$-uniform hypergraph setting, with a suitable probability of hash values becoming $1$. For triangles ($3$-uniform hypergraph), the particular hash function is thus $h:[k]^3 \rightarrow \{0,1\}$. So, there is no concept of proper and improper edges. All edges are taken into consideration and then all triangles, and we can work with any number of triangles incident on an edge. Moreover, there is an inductive nature to this particular choice of hash function. Consider fixing the color of a vertex of a triangle. Then the hash function $h : [k]\times[k]\rightarrow \{0,1\}$ behaves as a random coloring on the other two vertices of the triangle. 

\subsection{Proof of the lemma}
\begin{proof} We prove this lemma using induction on $d$.

\paragraph*{The base case:} For $d=1$, the hash function is $h_1:[k]\rightarrow \{0,1\}$. The vertices in $A_1$ are \colored with $[k]$. So, $\cR_1=\sum\limits_{c_1 \in [k]}m_o(\chi(1,c_1)) \times h_1(c_1)$. We have, $\E[\cR_1]=\frac{m_o(A_1)}{k}$ \remove{(Gopi: what does $m_o(A_1)$ mean? Is it just $\size{A_1}$? Pls discuss.)}. Note that the set of $1$-uniform hyperedges $\cF_o(A_1)$ is a subset of $A_1$. For each hyperedge $F \in \cF_o(A_1)$, let $X_F$ be the indicator random variable such that $X_F=1$ if and only if $h_1(x)=1$, where the only element in $F$ gets color $x$. Therefore, $\pr(X_F=1)=\frac{1}{k}$. Observe that the random variables $X_F$ are independent. Then, $\cR_1=\sum\limits_{F \in \cF_o(A_1)}X_F$. Now, applying Hoeffding's bound (see Lemma~\ref{lem:hoeff_inq} in Appendix~\ref{sec:prelim}), we get 
$$\pr\left(\size{\cR_1-\frac{m_o(A_1)}{k}} \geq  4{\theta} \cdot \sqrt{ \log n \cdot m_o(A_1)}\right) \leq \frac{2}{n^{16\theta^2}}\leq  p_1.$$

\paragraph*{The inductive case:} Let $A=\{1,\ldots,n'\}$, where $n' \leq n$. Let $Z_i \in [k]$ be the random variable that denotes the color assigned to the vertex $i \in [n'] $. Note that $\cR_d$ is a function of $Z_1,\ldots,Z_{n'}$, that is, $\cR_d=f(Z_1,\ldots,Z_{n'})$. Now, consider the following definition.

\begin{defi} \label{defi:proper-order} An ordered hyperedge $F_o \in \cF_o\left(\chi({1,c_1}),\ldots,\chi({1,c_{a_1}}),\ldots, \chi({s,c_{d-a_s+1}})\ldots \chi({s,c_d})\right)$ is said to be \emph{properly} colored if there exists $(c_1,\ldots,c_d) \in [k]^d$ such that $h_d(c_1,\ldots,c_d)=1$. \end{defi}

\noindent Observe that the probability that a hyperedge is properly colored is $\frac{1}{k}$, and $\cR_d$, as defined in the statement of Lemma~\ref{lem:sparse2}, represents the number of properly colored hyperedges. So, $\E[\cR_d]=m_o(A_1^{[a_1]},\ldots,A_s^{[a_s]})/k$. 

	Let us focus on the instance when vertices (in some arbitrary order) $1,\ldots,t-1$ have been colored and we are going to color the vertex $t$. Recall that $\cF_o(t)$ denotes the set of ordered hyperedges containing $t$ as one of the vertices. Let $\cF_o(t,\mu) \subseteq \cF_o(t)$ be the set of ordered hyperedges containing $t$ as the $\mu$-th vertex, where $\mu \in [d]$. Note the following observation that will be used later in the proof. 

\begin{obs} \label{obs:mu-edge} $\size{\cF_o(t,\mu)}=\frac{\size{\cF_o(t)}}{d},$ where $\mu \in [d]$. \end{obs}

A hyperedge $F \in \cF_o(t,\mu)$ is said to be of type $\lambda$ if $F$ has exactly $\lambda$ many vertices from $[t]$, where $\lambda \in [d]$. For $\mu \in [d]$, let $\cF_o^\lambda(t)$ and $\cF_o^\lambda (t,\mu)$ be the set of type $\lambda$ ordered hyperedges in $\cF_o(t)$ and $\cF_o(t,\mu)$, respectively. Given that the vertex $t$ is colored with color $c \in [k]$, let $N_c^\lambda(t)$ and $N_c^\lambda(t,\mu)$ be the random variables that denote the number of ordered hyperedges in $\cF_o^\lambda(t)$ and $\cF_o^\lambda(t,\mu)$ that are properly colored, respectively. 
 
Let $\E_{\cR_d}^t$ denote the difference in the conditional expectation of the number of ordered hyperedges that are properly colored such that the $t$-th vertex is differently colored by considering the hyperedges in each $\cF_o^\lambda (t,\mu), \lambda,\mu \in [d]$, separately.

\begin{eqnarray*}
	\E_{\cR_d}^t &=& \size{\E[\cR_d~|~Z_1,\ldots,Z_{t-1}, Z_t =\rho] - \E[\cR_d~|~Z_1,\ldots,Z_{t-1}, Z_t =\nu]}\\
	&=& \size{\sum\limits_{\mu=1}^d\left(N_{\rho}^{d}(t,\mu) - N_{\nu}^{d}(t,\mu)\right) + \sum\limits_{\lambda=1}^{d-1} \E\left[ N_{\rho}^{\lambda}(t) - N_{\nu}^{\lambda}(t) \right]}\\
	&\leq& \sum\limits_{\mu=1}^d \size{N_{\rho}^{d}(t,\mu) - N_{\nu}^{d}(t,\mu)} + \size{\sum\limits_{\lambda=1}^{d-1} \E\left[ N_{\rho}^{\lambda}(t) - N_{\nu}^{\lambda}(t) \right]}
\end{eqnarray*}

Now, consider the following claim. 
\begin{cl} \label{cl:inter_sparse}
	\begin{itemize}
		\item[(a)] $\pr\left(\size{ N_{\rho}^{d}(t,\mu) - N_{\nu}^{d}(t,\mu)} \leq 2^{2d-1}\theta^{d-1} \sqrt{{(d-1)!\cF_o(t,\mu)}\log^{d-1} n}\right) \geq 1-2p_{d-1}$, where $\mu \in [d]$ and $\theta > d$ is the constant mentioned in the statement of Lemma~\ref{lem:sparse2}. 
		\item[(b)] $\E\left[ N_{\rho}^{\lambda}(t) - N_{\nu}^{\lambda}(t) \right]=0,~\forall\lambda \in [d-1]$.
	\end{itemize}
\end{cl}

\begin{proof}[Proof of Claim~\ref{cl:inter_sparse}]
\begin{itemize}
	\item[(a)] For simplicity, we argue for $\mu =1$. However, the argument will be similar for any $\mu\in [d]$. Consider a $(d-1)$-uniform hypergraph $\cH'$ such that $U(\cH')=[t-1]$ and $\cF(\cH')=\{(x_1,\ldots,x_{d-1}):(t,x_1,\ldots,x_{d-1})\in \cF_o^d(t,1)\}$. Let $h_{d-1}:[k]^{d-1} \rightarrow \{0,1\}$ be a hash function such that $h_{d-1}(x_1,\ldots,x_{d-1})=h_d(t,x_1,\ldots,x_d)$. Observe that $\pr(h_{d-1}({\bf a})=1) =\frac{1}{k}$ for each tuple ${\bf a} \in [k]^{d-1}$. Consider the $(d-1)$-partite hypergraph $\cH'\left(B_1^{[a_1-1]}B_2^{[a_2]}\cdots B_s^{[a_s]}\right)$, where $B_i = A_i \cap [t-1]$. Observe that $m_o\left(B_1^{[a_1-1]}B_2^{[a_2]}\cdots B_s^{[a_s]}\right)=\size{\cF_o^d(t,1)}$. Recall that the vertices in $\cup_{i=1}^s A_i$ are \colored with $[k]$. Let $\chi'(i,j)=\{v \in B_i:\mbox{$v$ is colored with color $j$}\}$. Let
		$$\cR_{d-1}=\sum\limits_{(c_1,\ldots,c_{d-1}) \in [k]^{d-1}} m_o\left(\chi'({1,c_1}),\ldots,\chi'({1,c_{a_1-1}}),\ldots \ldots, \chi'({s,c_{d-a_s}})\ldots \chi'({s,c_{d-1}})\right)\times h(c_1,\ldots,c_{d-1})$$

Observe that the random variables $N_\rho^d(t,1)$ and $N_\nu^d(t,1)$ follow the distribution of random variable $\cR_{d-1}$. By the induction hypothesis, 
$$\pr\left(\size{\cR_{d-1}-\frac{m_o\left(B_1^{[a_1-1]}B_2^{[a_2]}\cdots B_s^{[a_s]}\right)}{k}} \geq  2^{2d-2}\theta^{d-1} \sqrt{{(d-1)! m_o\left(B_1^{[a_1-1]}B_2^{[a_2]}\cdots B_s^{[a_s]}\right)}\log^{d-1} n}\right) \leq p_{d-1}$$

Using $m_o\left(B_1^{[a_1-1]}B_2^{[a_2]}\cdots B_s^{[a_s]}\right)=\size{\cF_o^d(t,1)}$, we have
 \begin{equation}
  \label{eqn:ind}
	 \pr\left(\size{\cR_{d-1} - \frac{\size{\cF_o^d(t,1)}}{k}} \geq  2^{2d-2}\theta^{d-1} \sqrt{{(d-1)!\cF^d_o(t,1)}\log^{d-1} n}\right) \leq p_{d-1}
 \end{equation}

Using the above equation, the claim follows using the following. 

		$$\pr\left(\size{ N_{\rho}^{d}(t,1) - N_{\nu}^{d}(t,1)} \geq 2^{2d-1}\theta^{d-1} \sqrt{{(d-1)!\cF_o(t,1)}\log^{d-1} n}\right) \leq 2 p_{d-1}.$$ 
		
Let $L=2^{2d-1}\theta^{d-1} \sqrt{{(d-1)!\cF_o(t,1)}\log^{d-1} n}$. Now, 
\begin{eqnarray*}
	&&\pr\left(\size{ N_{\rho}^{d}(t,1) - N_{\nu}^{d}(t,1)} \geq 2^{2d-1}\theta^{d-1} \sqrt{{(d-1)!\cF_o(t,1)}\log^{d-1} n}\right) \\
	&\leq& \pr\left(\size{N_{\rho}^{d}(t,1) -\frac{\size{\cF_o^d(t,1)}}{k}} \geq \frac{L}{2} \right) + \pr\left(\size{N_{\nu}^{d}(t,1) -\frac{\size{\cF_o^d(t,1)}}{k}} \geq \frac{L}{2} \right)\\
	&=& 2 \cdot \pr\left(\size{N_{\rho}^{d}(t,1) -\frac{\size{\cF_o^d(t,1)}}{k}} \geq  2^{2d-2}\theta^{d-1} \sqrt{{(d-1)!\cF_o(t,1)}\log^{d-1} n}\right) \\
	&\leq& 2\cdot \pr\left(\size{N_{\rho}^{d}(t,1) -\frac{\size{\cF_o^d(t,1)}}{k}} \geq 2^{2d-2}\theta^{d-1} \sqrt{{(d-1)!\cF^d_o(t,1)}\log^{d-1} n}\right)~~~~(\because \cF_o^d(t,1) \leq \cF_o(t,1))\\
	&\leq& 2p_{d-1}~~~~(\mbox{By Equation~\ref{eqn:ind}})
\end{eqnarray*}

\item[(b)]First, consider the case when $t$ is colored with color $\rho$. For $F \in \cF_o^\lambda(t), \lambda \in [d-1]$, let $X_F$ be the indicator random variable such that $X_F=1$ if and only if $F$ is properly colored. As $F$ is of type $\lambda$, there exists at least one vertex in $F$ that is not colored yet, that is, $\pr(X_F=1)=\frac{1}{k}$. Observe that $N^\lambda_\rho(t)=\sum\limits_{F \in \cF_o^\lambda(t)}X_F$. Hence, $\E\left[N^\lambda_\rho(t)\right]=\size{\cF_o^\lambda(t)}/k$. Similarly, one can show that $\E\left[N^\lambda_\nu(t)\right]=\size{\cF_o^\lambda(t)}/k$. Hence, $\E\left[ N_{\rho}^{\lambda}(t) - N_{\nu}^{\lambda}(t) \right]=0$.
\end{itemize}

\end{proof}

Now, let us come back to the proof of Lemma~\ref{lem:sparse2}. By Claim~\ref{cl:inter_sparse} and Observation~\ref{obs:mu-edge}, we have the following with probability at least $1-2d\cdot p_{d-1}$.
\begin{eqnarray*}
\E^t_{\cR_d} &\leq& 2^{2d-1}\theta^{d-1} \cdot d \cdot \sqrt{(d-1)!\frac{\cF_o(t)}{d} \log^{d-1} n} =  2^{2d-1}\theta^{d-1} \sqrt{{d!\cF_o(t)} \log^{d-1} n}=c_t,
\end{eqnarray*}
where $c_t= 2^{2d-1}{\theta}^{d-1} \sqrt{{d!\cF_o(t)} \log^{d-1} n}$. 

Let $\cB$ be the event that there exists $t \in [n]$ such that $\E_{\cR^d}^t > c_t$. By the union bound over all $t \in [n]$, $\pr(\cB) \leq 2dnp_{d-1}={2dn}\frac{{(d-1)!}}{n^{4\theta-2(d-1)}}\leq \frac{2d!}{n^{4\theta-2d+1}}$. Using the method of \emph{averaged bounded difference}~\cite{DubhashiP09} (See Lemma~\ref{lem:dp} in Appendix~\ref{sec:prelim}), we have 
$$ \pr\left(\size{\cR_d - \E[\cR_d]} >\delta + m_o(A_1^{[a_1]},\ldots,A_s^{[a_s]})\pr(\cB) \right) \leq e^{-{\delta^2}/{\sum\limits_{t=1}^{n} c_t^2}} + \pr(\cB)$$

We set $\delta = 2 \sqrt{\theta \log n \cdot \sum\limits_{t \in [n]}c_t^2}=2^{2d}\theta^{d-1/2}  \sqrt{d!m_o(\ssubset)\log ^d n}$. Using $m_o(\ssubset) \leq n^d$ and $\pr(\cB) \leq \frac{2d!}{n^{4\theta -4d+1}}$ , we have 
\begin{eqnarray*}
	&& \pr\left(\size{\cR_d-\frac{m_o(A_1^{[a_1]},\ldots,A_s^{[a_s]})}{k}} > 2^{2d} \theta^{d-1/2} \sqrt{d! m_o(A_1^{[a_1]},\ldots,A_s^{[a_s]}) \log ^d n}+\frac{2d!n^d}{n^{4\theta -2d+1}}\right) \\
	&\leq& \frac{1}{n^{4\theta}}+\frac{2d!}{n^{4\theta-2d+1}}
\end{eqnarray*}
Assuming $n \gg d$, 
$ \pr\left(\size{\cR_d - \frac{m_o(A_1^{[a_1]},\ldots,A_s^{[a_s]})}{k}} \geq 2^{2d} \theta ^d\sqrt{d! m_o(A_1^{[a_1]},\ldots,A_s^{[a_s]}) \log ^d n}  \right) \leq \frac{d!}{n^{4\theta -2d}}=p_{d} $
\end{proof}

\section{Proof of Lemma for Exact Estimation} \label{sec:exact}

\noindent In this Section, we prove Lemma~\ref{lem:exact}. We start with the description of the algorithm.
We restate the lemma for easy reference.
\begin{lem} [Exact Estimation: Lemma~\ref{lem:exact} restated] \label{lem:exact1} There exists a deterministic algorithm $\cAe$ that takes as input a $d$-uniform hypergraph $\cH$, constants $a_1,\ldots,a_s \in [d]$ such that $\sum_{i=1}^s a_i =d$ where $s \in [d]$, pairwise disjoint subsets $A_1,\ldots,A_s$ of $U(\cH)$, and threshold parameter $\tau \in \N$, and decides whether $m_o(\ssubset) \leq \tau$ using $\cO_d(\tau \log n)$ \gonepis queries. Moreover, the algorithm finds the exact value of $m_o(\ssubset)$ when $m_o(\ssubset) \leq \tau$. \end{lem}

\begin{proof} The algorithm first determines the exact value of $m_o(\ssubset)$ using $\Oh_d(m_o(\ssubset) \log n)$ \gonepis queries. Then, we show how to modify it to obtain the desired bound. We initialize a tree ${\cal T}$ with $(\ssubset)$ as the root. The nodes of the tree are labeled with either $0$ or $1$. If $m_o(\ssubset)=0$, we label the root with $0$ and terminate. Otherwise, we label the root with $1$, and as long as there is a leaf node $(B_1^{[b_1]},\ldots,B_{t}^{[b_t]})$ labeled with $1$, we do the following. Here, note that $0 \leq b_i \leq d$ and $\sum\limits_{i=1}^t b_i = d$. Also, for $i \neq j$, either $B_i \cap B_j = \emptyset$ or $B_i=B_j$.

\begin{itemize}
	\item[(i)] If $m_o(\tsubset)=0$ or there exists an $i \in [t]$ such that $\size{B_i} < b_i$, then we label $(\tsubset)$ with $0$. Otherwise, we label $(\tsubset)$ with $1$.
	\item[(ii)] We partition each $B_i$ into two parts $B_{i1}$ and $B_{i2}$ such that $\size{B_{i1}}=\lceil \frac{B_{i}}{2} \rceil$ and $\size{B_{i2}}=\lfloor \frac{B_{i}}{2} \rfloor$. There may exist some $B_i$ for which $B_{i2}=\emptyset$. We add nodes of the form $(C_{11j},\ldots,C_{1b_1j}, \ldots, C_{t1j},\ldots,C_{tb_tj})$ as the children of $(\tsubset)$, for $j\in \{1,2\}$ where $C_{ib_i1}=\lceil \frac{C_{ib_i}}{2}\rceil$ and $C_{ib_i2}=\lfloor \frac{C_{ib_i}}{2}\rfloor$. Note that for each node $(\tsubset)$ with label $1$, we add $\prod\limits_{i=1}^{t} 2^{b_i} = 2^d$ nodes as children of it.
\end{itemize}

Let $\cT'$ be the tree after deleting all the leaf nodes in $\cT$. Observe that $m_o(\tsubset)$ is the number of leaf nodes in $\cT'$ and
\begin{itemize}
	\item the height of ${\cal T}$ is bounded by $\max\limits_{i \in [t]} \log \size{A_i} +1  \leq \log n +1 $. 
	\item the query complexity of the above procedure is bounded by the number of nodes in ${\cal T}$ as we make at most one query per node of $\cT$.
\end{itemize}

The number of nodes in $\cT'$ which equals the number of internal nodes of $\cT$ is bounded by $(\log n+1) m_o(\tsubset)$. Hence, the number of leaf nodes in $\cT$ is at most $2^d (\log n + 1)m_o(\ssubset)$. The total number of nodes in $\cT$ is at most $(2^d+1) (\log n+1)m_o(\ssubset) \leq 2^{d+2}m_o(\ssubset) \log n$. Overall, the number of \gonepis queries made by our algorithm is at most $ 2^{d+2}m_o(\ssubset) \log n$.
 
The algorithm $\cAe$ for Lemma~\ref{lem:exact1} proceeds similar to the one presented above by initializing a tree $\cT$ with $(\ssubset)$ as the root. The pseudo-code for $\cAe$ is presented in Algorithm~\ref{algo:exact}. If $m_o(\ssubset) \leq \tau$, then we can find the exact value of $m_o(\ssubset)$ using at most $2^{d+2} \tau \log n$ \gonepis queries and the number of nodes in $\cT$ is bounded by $ 2^{d+2} \tau \log n$. So, if the number of nodes in $\cT$ is more than $2^{d+2}  m_o(\ssubset) \log n$ at any instance during the execution of the algorithm, we report $m_o(\ssubset) >\tau$ and terminate. Hence, our algorithm makes $\Oh_d(\tau \log n)$ many \gonepis queries, decides whether $m_o(\ssubset) \leq \tau$, and determines the exact value of $m_o(\ssubset)$ when $m_o(\ssubset) \leq \tau$. 
\end{proof}

\begin{algorithm}[h]

\SetAlgoLined

\caption{{$\cAe$($\dsubset,\tau$)}}
\label{algo:exact}
\KwIn{Subsets $A_1,\ldots, A_d$ of vertex set $U(\cH)$ of a $d$-uniform hypergraph $\cH$ and a threshold $\tau$.}
\KwOut{Check if $m_o(\ssubset) \leq \tau$. If Yes, return the exact value of $m_o(\ssubset)$.}

Initialize a tree ${\cal T}$ with $(\ssubset)$ as the root and label the root with $1$.\\
\If{the number of nodes in $\cT$ is more than $2^{d+2} \tau \log n$ }
{
Report that $m_o(\ssubset) > \tau$ as the output.
}
\While{there is a leaf node $m_o(\tsubset)=0$ in the tree with label $1$}
{
{Note that $0 \leq b_i \leq d$ and $\sum\limits_{i=1}^t b_i = d$. Also, for $i \neq j$, either $B_i \cap B_j = \emptyset$ or $B_i=B_j$.}\\
\If{ $m_o(\tsubset)=0$ or there exists an $i \in [t]$ such that $\size{B_i} < b_i$}
{ 
Label $(\tsubset)$ with $0$.
}
\Else
{
Label $(\tsubset)$ with $1$.\\
Partition each $B_i$ into two parts $B_{i1}$ and $B_{i2}$ such that $\size{B_{i1}}=\lceil \frac{B_{i}}{2} \rceil$ and $\size{B_{i2}}=\lfloor \frac{B_{i}}{2} \rfloor$. \\
Add nodes of the form $(C_{11j},\ldots,C_{1b_1j}, \ldots, C_{t1j},\ldots,C_{tb_tj})$ as children of $(\tsubset)$, for $j\in \{1,2\}$, where $C_{ib_i1}=\lceil \frac{C_{ib_i}}{2}\rceil$ and $C_{ib_i2}=\lfloor \frac{C_{ib_i}}{2}\rfloor$.
}
}
Generate $\cT'$ as the tree after deleting all the leaf nodes in $\cT$.\\
Report the number of leafs in $\cT'$ as $m_o(\ssubset)$ as the output.
\end{algorithm}

\section{Proof of Lemma for Coarse Estimation} \label{sec:coarse}

\noindent We now prove Lemma~\ref{lem:coarse_main}. The algorithm corresponding to Lemma~\ref{lem:coarse_main} is Algorithm~\ref{algo:coarse}. Algorithm~\ref{algo:verify} is a subroutine in Algorithm~\ref{algo:coarse}. Algorithm~\ref{algo:verify} determines whether a given estimate $\hat{R}$ is correct up to $\Oh_d(\log ^{2d-3} n)$ factor. The Lemma~\ref{lem:coarse1} and~\ref{lem:coarse2} are intermediate results needed to prove Lemma~\ref{lem:coarse_main}.

\begin{lem} \label{lem:coarse1} If $\hat{\cR} \geq 20d^{2d-3}4^d m_o(\dsubset) \log^{2d-3} n$, then 
	$$ \pr(\mbox{\verest($\dsubset,\hat{\cR}$) accepts}) \leq \frac{1}{20 \cdot 2^d}$$
\end{lem}

\begin{proof} Consider the set of ordered hyperedges $\cF_o(\dsubset)$ in $\cH(\dsubset)$. Let us construct the sets $B_{i,{\bf j}}$ for $1\leq i\leq d$ as given in Algorithm \ref{algo:verify}. For an ordered hyperedge $F_o \in \cF_o(\ssubset)$ and ${\bf j} \in \left[(d\log n)^*\right]^{d-1}$ \footnote{Recall that $[n]^*$ denotes the set $\{0,\ldots,n\}$}, let $X^{{\bf j}}_{F_o}$ denote the indicator random variable such that $X^{\bf j}_{F_o}=1$ if and only if $F_o \in \cF_o(B_{1,{\bf j}}, \ldots, B_{d,{\bf j}})$ and $X_{\bf j}=\sum\limits_{F_o \in \cF_o(\dsubset)} X^{\bf j}_{F_o}$. Note that $m_o(B_{1,{\bf j}}, \ldots, B_{d,{\bf j}}) = X_{\bf j}$. We have, 

$$\pr\left( X^{\bf j}_{F_o} =1\right) = \prod\limits_{i=1}^d (p(i,{\bf j})) \leq \frac{2^{j_1}}{\hat{\cR}} \cdot \frac{2^{j_2}}{2^{j_1}} d\log n \cdot \cdot \cdot \cdot \cdot \frac{2^{j_{d-1}}}{2^{j_{d-2}}} d\log n\cdot \frac{1}{2^{j_{d-1}}} = \frac{d^{d-2}\log ^{d-2} n}{\hat{\cR}}$$ 

\noindent Then, $\E\left[X_{\bf j}\right] \leq \frac{m_o(\dsubset)}{\hat{\cR}} d^{d-2}\log^{d-2} n$, and since $X_{\bf j} \geq 0$, $$\pr \left( X _{\bf j}\neq 0\right)=\pr(X_{\bf j} \geq 1) \leq \E \left[X_{\bf j} \right] \leq \frac{m_o(\dsubset)}{\hat{\cR}} d^{d-2}\log ^{d-2} n$$
Now, using the fact that $\hat{\cR} \geq 20d^{2d-3} \cdot 4^{d} \cdot m_o(\dsubset) \log^{2d-3} n$, we have 
$$ \pr\left(X_{\bf j} \neq 0 \right) \leq \frac{1}{20d^{d-1} \cdot 4^{d} \cdot \log^{d-1}  n}. $$

Observe that \verest accepts if and only if there exists ${\bf j} \in [(d\log n)^*]$ such that $X_{\bf j} \neq 0$. Using the union bound, we get
\begin{eqnarray*}
	\pr(\mbox{\verest$(\dsubset,\hat{\cR})$ accepts}) &\leq&  \sum\limits_{{\bf j} \in [(d\log n)^*]^{d-1}} \pr(X_{\bf j} \neq 0)\\ 
	&\leq& \frac{(d\log n +1)^{d-1}}{20 \cdot 4^d\cdot (d\log n)^{d-1}} \leq \frac{1}{20\cdot 2^d}.
\end{eqnarray*}

\end{proof}

\begin{algorithm}

\SetAlgoLined

\caption{\verest($\dsubset,\hat{\cR}$)}
\label{algo:verify}
\KwIn{$A_1,\ldots, A_d$ be subsets of vertex set $U(\cH)$ of a $d$-uniform hypergraph $\cH$ and an estimate $\hat{\cR}$}
\KwOut{{\sc Accept} or {\sc Reject} $\hat{\cR}$ as per stated in Lemma~\ref{lem:coarse1} and~\ref{lem:coarse2}.}

\For{$j_1= d\log n$ to $0$}
{
	\For{$j_2=d\log n$ to $0$}
	{
		\ldots \\
		\ldots\\
		\For{$j_{d-1}=d\log n$ to $0$}
		{
			Let ${\bf j}=(j_1,\ldots,j_{d-1}) \in [(d\log n)^*]^{d-1}$\\
			Let $p(1,{\bf{  j}}) =\min\{\frac{2^{j_1}}{\hat{\cR}},1\}$\\
			Let $p(i,{\bf j}) = \min\{2^{j_i - j_{i-1}} \cdot d\log n,1\}$, where $2 \leq i \leq d-1$\\
			Let $p(d,{\bf j}) = \min\{{2^{-j_{d-1}}},1\}$\\
			Construct $B_{i,{\bf j}} \subseteq A_i$ for each $i\in [d]$. Initialize $B_{i,{\bf j}}$ with $\phi$ for all $i$.\\
			Sample each element of $A_i$ independently with probability $p(i,{\bf j})$ and add it to $B_{i,{\bf j}}$.\\
			\If{$\left( m(B_{1,{\bf j}},\ldots, B_{d,{\bf j}}) \neq 0 \right)$}
				{\sc Accept}
		}
	}
}

{\sc Reject}
\end{algorithm}

\begin{lem} \label{lem:coarse2} If $\hat{\cR} \leq \frac{m_o(\dsubset)} {4d \log n}$, $\pr( \mbox{\verest ($\dsubset,\hat{\cR}$) accepts}) \geq \frac{1}{2^d}$. \end{lem}

\begin{proof} First, we define some quantities and prove Claim~\ref{clm:verify}. Then we will prove Lemma~\ref{lem:coarse2}. For $q_1 \in [(d \log n)^*]$, let $A_1(q_1) \subseteq A_1$ be the set of vertices in $A_1$ such that for each $u_1 \in A_1(q_1)$, the number of hyperedges in $\cF_o(\dsubset)$, containing $u_1$ as the first vertex, lies between $2^{q_1}$ and  $ 2^{q_1+1}-1$. For $2 \leq i \leq d-1$, and $q_j \in [(d \log n)^*]~\forall j \in [i-1]$, consider $u_1 \in A_1(q_1), u_2 \in A_2((q_1,u_1),q_2)\ldots,$ $ u_{i-1} \in A_{i-1}((q_1,u_1),\ldots,(q_{i-2}u_{i-2}),q_{i-1}) $. Let $A_i((q_1,u_1),\ldots, (q_{i-1},u_{i-1}),q_i)$ be the set of vertices in $A_i$ such that for each $u_i \in A_i((q_1,u_1),\ldots, (q_{i-1},u_{i-1}),q_i) $, the number of ordered hyperedges in $\cF_o(\dsubset)$, containing $u_j$ as the $j$-th vertex for all $j \in [i]$, lies between $2^{q_i}$ and  $ 2^{q_i+1}-1$. We need the following result to proceed further. For ease of presentation, we use $(Q_i,U_i)$ to denote $(q_1,u_1),\ldots, (q_{i-1},u_{i-1})$ for $ 2 \leq i \leq d-1$. Now, we prove the following claim. It will be required to prove the lemma.
\begin{cl} \label{clm:verify}
\begin{itemize}
	\item[(i)] There exists $q_1 \in [(d \log n)^*]$ such that $\size{A_1(q_1)} > \frac{m_o(\dsubset)}{2^{q_1+1}(d\log n +1)}$.
	\item[(ii)] Let $2 \leq i \leq d-1$ and $q_j \in [(d \log n)^*]~\forall j \in [i-1]$. Let $u_1 \in A_1(q_1)$, $u_{j} \in A_{j}((Q_{j-1},U_{j-1}),q_{j})$ $\forall j \neq 1$ and $j<i$. There exists $q_i \in [(d \log n)^*]$ such that $\size{ A_i((Q_i,U_i),q_i)} > \frac{2^{q_{i-1}}}{2^{q_i +1}(d\log n +1)}$.
\end{itemize}
\end{cl}

\begin{proof}
\begin{itemize}
	\item[(i)] Observe that $m_o(\dsubset)=\sum\limits_{q_1=0}^{d\log n } m_o(A_1(q_1),A_2,\ldots,A_d)$. So, 
there exists $q_1 \in [(d \log n)^*]$ such that $m_o(A_1(q_1),A_2,\ldots,A_d) \geq \frac{m_o(\dsubset)}{d\log n + 1}$. From the definition of $A_1(q_1)$, $ m_o(A_1(q_1),A_2,\ldots,A_d) < \size{A_1(q_1)} \cdot 2^{q_1 +1}$. Hence, there exists $q_1 \in [(d \log n)^*]$ such that 
$$ 
	\size{A_1(q_1)} > \frac{m_o(A_1(q_1),A_2,\ldots,A_d)}{2^{q_1+1}} 
	\geq  \frac{m_o(\dsubset)}{2^{q_1+1} (d\log n+1)}.
$$

	\item[(ii)] Note that $m_o(\{u_1\},\ldots, \{u_{i-1}\},A_i,\ldots,A_d)=\sum_{q_i=0}^{d\log n } m_o(\{u_1\},\ldots, \{u_{i-1}\},A_i((Q_{i-1},U_{i-1}),q_i),\ldots,A_d)$. So, there exists $q_i \in [(d \log n)^*]$ such that 

$$ m_o(\{u_1\},\ldots, \{u_{i-1}\},A_i((Q_{i-1},U_{i-1}),q_i),\ldots,A_d) \geq \frac{m_o(\{u_1\},\ldots, \{u_{i-1}\},A_i,\ldots,A_d)}{d\log n +1}$$ 

From the definition of $A_i((Q_{i-1},U_{i-1}),q_i)$, we have $$m_o(\{u_1\},\ldots, \{u_{i-1}\},A_i((Q_{i-1},U_{i-1}),q_i),\ldots,A_d)  < \size{A_i((Q_{i-1},U_i),q_i)} \cdot 2^{q_i +1}$$ Hence, there exists $q_i \in [(d \log n)^*]$ such that 

\begin{eqnarray*}
	\size{A_i((Q_{i-1},U_i),q_i)} &>& \frac{m_o(\{u_1\},\ldots, \{u_{i-1}\},A_i((Q_{i-1},U_{i-1}),q_i),\ldots,A_d\})}{2^{q_i +1}}\\
	&\geq& \frac{m_o(\{u_1\},\ldots, \{u_{i-1}\},A_i,\ldots,A_d\})}{2^{q_i +1}(d\log n +1)} \geq \frac{2^{q_{i-1}}}{2^{q_i +1}(d\log n +1)}
\end{eqnarray*}
\end{itemize}
\end{proof}

We will be done with the proof of Lemma~\ref{lem:coarse2} by showing the following. \verest accepts with probability at least $1/5$ when the loop variables $j_1,\ldots,j_{d-1}$ respectively attain values $q_1,\ldots,q_{d-1}$ such that $\size{A_1(q_1)} > \frac{m_o(\dsubset)}{2^{q_1+1}(d\log n +1)}$ and $\size{ A_i((Q_i,U_i),q_i)} > \frac{2^{q_{i-1}}}{2^{q_i +1}(d\log n +1)}~ \forall i \in [d-1] \setminus \{1\}$. The existence of such $j_i$'s is evident from Claim~\ref{clm:verify}. Let ${\bf q}=(q_1,\ldots,q_{d-1})$. Recall that $B_{i,{\bf q}} \subseteq A_i$ is the sample obtained when the loop variables $j_1,\ldots,j_{d-1}$ attain values $q_1,\ldots,q_{d-1}$, respectively. Let $\cE_i, i \in [d-1],$ be the events defined as follows.

\begin{itemize}
	\item $\cE_1~:~A_1(q_1) \cap B_{1,{\bf q}} \neq \emptyset$.
	\item $\cE_i~:~A_j((Q_{j-1},U_{j-1}),q_j) \cap B_{j,{\bf q}} \neq \emptyset$, where $2\leq i \leq d-1$.
\end{itemize} 
Observe that 
$$\pr(\overline{\cE_1}) \leq \left( 1- \frac{2^{q_1}}{\hat{\cR}}\right)^{\size{A_1(q_1)}} \leq \exp {\left(-\frac{2^{q_1}}{\hat{\cR}}\size{A_1(q_1)}\right)} \leq \exp{\left(-\frac{2^{q_1}}{\hat{\cR}}\cdot \frac{m_o(\dsubset)}{2^{q_1+1}(d \log n +1)}\right)} \leq \exp{(-1)}$$  

The last inequality uses the fact that $\hat{\cR} \leq \frac{m_o(\dsubset)} {4d \log n}$, from the condition of the lemma. Assume that $\cE_1$ occurs and $u_1 \in A_1(q_1) \cap B_{1,{\bf q}} $. We will bound the probability that $A_2(Q_1,U_1),q_2) \cap A_{2,{\bf q}} = \emptyset$, that is $\overline{\cE_2}$. Note that, by Claim~\ref{clm:verify} (ii), $ \size{A_2(Q_1,U_1),q_2)} \geq \frac{2^{q_{1}}}{2^{q_2 +1}(d\log n +1)}$. So,
$$ \pr\left(\overline{\cE_2}~|~\cE_1\right) \leq \left(1-\frac{2^{q_2}}{2^{q_1}} \log n \right)^{\size{A_2(Q_1,U_1),q_2)}} \leq \exp{(-1)}$$

Assume that $\cE_1,\ldots,\cE_{i-1}$ holds, where $3 \leq i \in [d-1]$. Let $u_1 \in A_1(q_1)$ and $u_{i-1} \in A_{i-1}((Q_{i-2},U_{i-2}),q_{i-1})$. We will bound the probability that $A_i((Q_{i-1},U_{i-1}),q_i) \cap B_{i,{\bf q}} = \emptyset$, that is $\overline{\cE_i}$. Note that $\size{A_i((Q_{i-1},U_{i-1}),q_i)} \geq \frac{2^{q_{i-1}}}{2^{q_{i} +1}(d\log n +1)}$. So, for $3 \leq i \in [d-1]$,

$$ \pr\left( \overline{\cE_i}~|~\cE_1,\ldots,\cE_{i-1}\right) \leq \left(1-\frac{2^{q_i}}{2^{q_{i-1}}} \log n \right)^{\size{A_i(Q_{i-1},U_{i-1}),q_i)}} \leq  \exp{(-1)}$$

Assume that $\cE_1,\ldots,\cE_{d-1}$ holds. Let $u_1 \in A_1(q_1)$ and $u_{i-1} \in A_{i-1}((Q_{i-2},U_{i-2}),q_{i-1})$ for all $i \in [d] \setminus \{1\}$. Let $S \subseteq A_d$ be the set of $d$-th vertex of the ordered hyperedges in $\cF_o(\dsubset)$ having $u_j$ as the $j$-th vertex for all $j \in [d-1]$. Note that $\size{S} \geq 2^{q_{d-1}}$. Let $\cE_d$ be the event that represents the fact $S \cap B_{d,{\bf q}} \neq \emptyset$. So,

$$ \pr(\overline{\cE_d}~|~\cE_1,\ldots,\cE_{d-1}) \leq \left(1- \frac{1}{2^{q_{d-1}}} \right)^{q_{d-1}} \leq \exp{(-1)}$$ 

Observe that \verest accepts if $m(A_{B,{\bf q}},\ldots,B_{d,{\bf q}}) \neq 0$. Also, $m(B_{1,{\bf q}},\ldots,B_{d,{\bf q}}) \neq 0$ if $\bigcap\limits_{i=1}^d \cE_i$ occurs. Hence,

\begin{eqnarray*}
	\pr(\mbox{\verest$(\dsubset, \hat{\cR})$ accepts}) &\geq& \pr\left( \bigcap\limits_{i=1}^d \cE_i \right)\\
	&=& \pr(\cE_1)\prod\limits_{i=2}^d \pr(\cE_i~|~\bigcap\limits_{j=1}^{i-1}\cE_j)\\
	&>& (1-\exp{(-1)})^d \geq \frac{1}{2^d}
\end{eqnarray*}

\end{proof}

\begin{algorithm}[H]
\caption{\cest($\dsubset$)}
\label{algo:coarse}
\KwIn{$d$ subsets $\dsubset \subseteq U(\cH)$.}
\KwOut{An estimate $\hat{E}$ for $m_o(\dsubset)$ as per stated in Lemma~\ref{lem:coarse_main1}.}
\For{$~\hat{\cR}= n^d,n^d/2,\ldots, 1$}
{
	Repeat \verest $(\dsubset,\hat{\cR})$ for $\Gamma=d~4^d~2000~\log n $ times.\\
	If more than $\frac{\Gamma}{10~2^d}$ runs of \verest Accepts, then output ${\hat{E}}=\frac{\hat{\cR}}{d^{d-2}\cdot 2^d}$.
}
\end{algorithm}

Now, we will prove Lemma~\ref{lem:coarse_main}. We restate the lemma for easy reference.
\begin{lem}[Coarse estimation~:~Lemma~\ref{lem:coarse_main} restated] \label{lem:coarse_main1} There exists an algorithm $\cAc$ that takes as input $d$ many subsets $A_1,\ldots, A_d$ of the vertex set $U(\cH)$ of a $d$-uniform hypergraph $\cH$. The algorithm $\cAc$ returns $\hat{E}$ as an estimate for $m_o(\dsubset)$ such that 
$$ \frac{m_o(\dsubset)}{8d^{d-1}2^d \log^{d-1} n} \leq \hat{E} \leq 20d^{d-1}2^d \cdot m_o(\dsubset) \log^{d-1} n$$ with probability $1-n^{-8d}$. Moreover, the number of \gtwopis queries made by the algorithm is $\Oh_d(\log ^{d+1} n)$. \end{lem}

\begin{proof} Note that an execution of \cest for a particular $\hat{\cR}$, repeats \verest for $\Gamma =d \cdot 4^d \cdot 2000 \log n$ times and gives output $\hat{\cR}$ if more than $\frac{\Gamma}{10 \cdot 2^d}$ many \verest accepts. For a particular $\hat{\cR}$, let $X_i$ be the indicator random variable such that $X_i=1$ if and only if the  $i$-th execution of \verest accepts. Also take $X=\sum_{i=1}^\Gamma X_i$. \cest gives output $\hat{\cR}$ if $X > \frac{\Gamma}{10 \cdot 2^d}$.

Consider the execution of \cest for a particular $\hat{\cR}$. If $\hat{\cR}  \geq 20d^{2d-3}4^d\cdot m_o(\dsubset)\cdot$ $\log ^{2d-3} n$, we first show that \cest does not accept with high probability. Recall Lemma~\ref{lem:coarse1}. If $\hat{\cR} \geq 20d^{2d-3}4^d \cdot m_o(\dsubset)\log ^{2d-3} n$, $\pr(X_i =1) \leq \frac{1}{20 \cdot 2^d}$ and hence $\E[X] \leq \frac{\Gamma}{20\cdot 2^d}$. By using Chernoff-Hoeffding's inequality~(See Lemma~\ref{lem:cher_bound2}~(i) in Section~\ref{sec:prelim}), 

$$ \pr \left(X > \frac{\Gamma}{10 \cdot 2^d} \right) =\pr\left( X > \frac{\Gamma}{20 \cdot 2^d} + \frac{\Gamma}{20 \cdot 2^d}\right) \leq \frac{1}{n^{10d}}$$ 

Using the union bound for all $\hat{\cR}$, the probability that \cest outputs some $\hat{E}=\frac{\hat{\cR}}{d^{d-2}\cdot 2^d}$ such that $\hat{\cR} \geq 20d^{2d-3}4^d \cdot m_o(\dsubset)\log^{2d-3}n$, is at most $\frac{d \log n}{n^{10}}$. Now consider the instance when the for loop in \cest executes for a $\hat{\cR}$ such that $\hat{\cR} \leq \frac{m_o(\dsubset)}{ 4d \log  n}$. In this situation, $\pr(X_i=1) \geq \frac{1}{2^d}$. So, $\E[X] \geq \frac{\Gamma}{2^d}$. By using Chernoff-Hoeffding's inequality~(See Lemma~\ref{lem:cher_bound2}~(ii) in Section~\ref{sec:prelim}), 

$$ \pr\left(X \leq \frac{\Gamma}{10 \cdot 2^d } \right) \leq \pr\left(X < \frac{\Gamma}{2^d} -\frac{4}{5} \cdot \frac{\Gamma}{ 2^{d}} \right) 	\leq \frac{1}{{n^{100d}}}$$  

By using the union bound for all $\hat{\cR}$, the probability that \cest outputs some $\hat{E}=\frac{\hat{\cR}}{d^{d-2}\cdot 2^d}$ such that $\hat{\cR} \leq \frac{m_o(\dsubset)}{ 4d \log  n}$, is at most $\frac{d \log n}{n^{100d}}$. Observe that, the probability that \cest outputs some $\hat{E}=\frac{\hat{\cR}}{d^{d-2}\cdot 2^d}$ such that $\hat{\cR}\geq d^{2d-3}4^d m_o(\dsubset)\log ^{2d-3} n$ or $\hat{\cR} \leq \frac{m_o(\dsubset)}{4d \log n}$, is at most $\frac{d\log n}{n^{10d}} +\frac{d\log n}{n^{100d}} \leq \frac{1}{n^{8d}}$. Putting everything together, \cest gives some $\hat{E}=\frac{\hat{\cR}}{d^{d-2} \cdot 2^d}$ as the output with probability at least $1-\frac{1}{n^{8d}}$ satisfying 

$$ \frac{m_o(\dsubset)}{8d^{d-1}2^d \log^{d-1} n} \leq \hat{E}=\frac{\hat{\cR}}{d^{d-2} \cdot 2^d} \leq 20d^{d-1}2^d \cdot m_o(\dsubset) \log^{d-1} n$$ 

From the description of \verest and \cest, the query complexity of \verest is $\Oh(\log ^{d-1} n)$ and \cest calls \verest $\Oh_d(\log n)$ times for each choice of $\hat{R}$. Hence, \cest makes $\Oh_d(\log ^{d+1} n)$ many \gpis queries. 
\end{proof}

\section{Algorithm} \label{sec:algo}

\noindent In this Section, we describe our $(1 \pm \eps)$-approximation algorithm for the hyperedge estimation problem in hypergraph $\cH$. When $\eps \leq \lbeps$, we compute the exact value of $m_o(\cH)$ by querying $m_o(\{a_1\},\ldots,\{a_d\})$ for all distinct $a_1,\ldots,a_d\in U(\cH)$, and this requires only polylogarithmic number of queries. When $\eps > \lbeps$, we do the following. We build a data structure $\cal D$ that maintains the following two quantities.
\begin{itemize}
	\item[(i)] An accumulator $\Psi$ that builds an estimate for the number of hyperedges. We initialize $\Psi=0$.
	\item[(ii)] A set of tuples $(A_{11},\ldots,A_{1d},w_1),\ldots,$ $(A_{\zeta 1},\ldots,A_{\zeta d},w_{\zeta})$ for some $\zeta>0$, where tuple $(A_{i1},\ldots,A_{id})$ corresponds to the $d$-partite sub-hypergraph $\cH(A_{i1},\ldots,A_{id})$ and $w_i$ is the weight associated to $\cH(A_{i1},\ldots,A_{id})$. 
\end{itemize}

The data structure $\cal D$ is initialized with $\Psi=0$, and only one tuple $(U^{[d]},1)$. The Algorithm performs the following steps.

\begin{itemize}
	\item[(1)] Checks whether there are any tuples left in $\cal D$. If some tuples are present in $\cal D$, it goes to Step 2. Else, it outputs $\Psi$ as the estimate for the number of hyperedges.

	\item[(2)]{\bf (Exact Counting)} Fixes the threshold $\tau=\frac{k^2 4^{2d}\theta^{2d} 16d^2 d! \log ^{d+2} n}{\eps^2}$, for $k=4$\footnote{We take $k=4$ for these calculations. The argument works for any $k$.}. For each tuple $(A_1,\ldots,A_d,w)$ in $\cal D$, it decides whether $m_o(\dsubset) \leq \tau$ using Lemma~\ref{lem:exact}. If $m_o(\dsubset) \leq \tau$, it adds $w m_o(\dsubset)$ to $\Psi$ and removes $(A_1,\ldots,A_d,w)$ from $\cal D$. If there are no tuples left in $\cal D$, it goes back to Step $1$. If the number of tuples is at most $\cN=\kappa_d \cdot \frac{\log ^{4d} n}{\eps^2}$, then it goes to Step $3$, else it goes to Step $4$. Note that $\kappa_d$ is a constant to be fixed later. By Lemma~\ref{lem:exact}, for each tuple this step requires $\cO_d(\tau \log n)=\cO_d\left( \frac{\log^{d+3}n}{\eps^2}\right)$ \gonepis queries.

	\item[(3)] {\bf (Sparsification)} For any tuple $(A_1^{[a_1]},\ldots,A_s^{[a_s]},w)$ in this step, we will have $m_o(A_1^{[a_1]},\ldots,A_s^{[a_s]})>\tau$. Recall that $A_i$ and $A_j$ are pairwise disjoint subsets for each $1\leq i <j \leq s$. It takes a hash function $h_d:[k]^{d}\rightarrow \{0,1\}$ such that $\pr\left(h_d({\bf a}=1)\right)=1/k$ independently for each tuple ${\bf a} \in [k]^d$. The vertices in $A=\bigcup\limits_{i=1}^s A_i$ are colored independently and uniformly at random with colors in $[k]$. It constructs sets $\chi(i,j)=\{v \in A_i:v~\mbox{is colored with color $j$}\}$ for $i \in [s]$ and $j \in [k]$. It adds each tuple in $\left(\chi({1,c_1}),\ldots,\chi({1,c_{a_1}}),\ldots, \chi({s,c_{d-a_s+1}}),\ldots,\chi({s,c_{d}})\right)$ with weight $4w$ to $\cal D$ for which $h_d(c_1,\ldots,c_d)=1$. It removes the tuple $(A_1^{[a_1]},\ldots,A_s^{[a_s]},w)$ from $\cal D$. After processing all the tuples in this step, it goes to Step $2$. Note that no query is required in this step. The constant $4$ is obtained by putting $k=4$ in Lemma~\ref{lem:sparse}.

	\item[(4)] ({\bf Coarse Estimation}) We have $r>\cN=\frac{\kappa_d\log^{4d} n}{\eps^2}$ tuples $\{(A_{i1},\ldots,A_{id},w_i):i \in [r]\}$ stored in the data structure $\cal D$. For each such tuple $(A_{i1},\ldots,A_{id},w_i)$, it finds an estimate $\hat{E_i}$ such that $\frac{m_o(A_1,\ldots,A_d)}{8 2^{d}d^{d-1} \log^{d-1} n} \leq \hat{E_i} \leq 20 2^{d}d^{d-1} \log^{d-1} n~m_o(A_1,\ldots,A_d)$. Using Lemma~\ref{lem:coarse_main}, this can be done with $\Oh_d(\log^{d+1} n)$ \gtwopis queries per tuple. It takes a sample from the set of tuples such that the sample maintains the required estimate \emph{approximately} using Lemma~\ref{lem:importance-app}. It uses the algorithm for Lemma~\ref{lem:importance-app} with $\lambda =\frac{\eps}{4d \log n}$, $\alpha =  20 2^{d}d^{d-1} \log^{d-1} n$ and $\delta = \frac{1}{n^{6d}}$ to find a new set $\{(A'_{i1},\ldots,A'_{id},w'_i):i \in [r']\}$ of tuples satisfying the following. We will have $\size{S-\sum_{i=1}^{r'} w'_i m_o(A'_{i1},\ldots,A'_{id})} \leq \lambda S$ with probability $1-\frac{1}{n^{6d}}$, where $S=\sum_{i=1}^r w_i m_o(A_{i1},\ldots,A_{id})$. Here, $r'\leq \kappa_d \cdot \frac{\log ^{4d} n}{\eps^2}$, where $\kappa_d$ is the constant mentioned in Step $2$. It removes the set of $r$ tuples $r > \cN$ from $\cal D$ and adds the set of $r'$ tuples. Since the algorithm of Lemma~\ref{lem:importance-app} does not require any queries, the number of \gtwopis queries in this step in each iteration is $\Oh_d(\log ^{d+1} n)$ per tuple. 

\end{itemize}

\section{Proof of Correctness} \label{sec:correct}
\noindent
We start with the following observation for the proof of correctness.

\begin{obs} \label{obs:num_tuples} There are at most $4^d\cdot \cN =4^d \kappa_d \cdot \frac{\log ^{4d} n}{\eps^2}$ tuples in the data structure $\cal D$ during the execution of the algorithm. \end{obs}

\begin{proof} The number of tuples in $\cal D$ can increase by a factor of $4^d$ in the sparsification step. Note that we apply the sparsification step only when there are at most $\cN=\kappa_2 \cdot \frac{\log ^{4d} n}{\eps^2}$ many tuples in $\cal D$. Hence, the number of tuples in $\cal D$ is at most $4^d \cdot \cN$. \end{proof}

Now we prove Lemma~\ref{lem:prob1}. We restate the lemma for easy reference.
\begin{lem}[Lemma~\ref{lem:prob1} restated] \label{lem:prob} If $\eps \geq \lbeps$, our algorithm produces $(1\pm \eps)$-approximation to $m_o(\cH)$ with probability at least $1-\frac{1}{n^{4d}}$ and makes $\Oh\left( \frac{\log ^{5d+4} n}{\eps^4}\right)$ queries, where each query is either a \gonepis query or a \gtwopis query. \end{lem}

To prove the above lemma, we need the following definition~\ref{defi:est} along with Observations~\ref{lem:iter1} and~\ref{lem:iter2}.

\begin{defi} \label{defi:est} Let $\mbox{{\sc Tuple}}_i$ be the set of tuples in the data structure $\cal D$ at the end of the $i$-th iteration. We have $\tuple_i^{\leq \tau}=\{(A_1,\ldots,A_d,w):m_o(\dsubset) \leq \tau\}$ and $\tuple_i^{>\tau}=\tuple_{i} \setminus \tuple_i^{\leq \tau}$. Let $\Psi_i$ denote the value of $\Psi$ after the $i$-th iteration where $i \in \N$. The estimate for $m_o(\cH)=m_o(U^{[d]})$ after the $i$-th iteration is given as $\est_{i}=\Psi_i + \sum\limits_{(A_1,\ldots,A_d,w) \in \tuple_i} w \cdot m_o(\dsubset)$. The number of active hyperedges after the $i$-th iteration is given as $\act_{i}=\sum\limits_{(A_1,\ldots,A_d,w) \in \tuple_i} m_o(\dsubset)$. \end{defi}

Note that if there are some tuples left in $\cal D$ at the end of the $i$-th iteration, we do not know the value of $\est_i$ and $\act_i$. However, we know $\Psi_i$. Observe that $\Psi_{0}=0$ and $\est_{0}=\act_{0}=m_o(\cH)$.

\begin{obs} \label{lem:iter1} Let there be only one tuple in the data structure $\cal D$ after the $i$-th iteration for any non-negative integer $i$. Then, $\est_{{i+1}}$ is a $(1+\lambda)$-approximation to $\est_i$, where $\lambda=\frac{\eps}{4d \log n}$, with probability at least $1-\frac{1}{n^{5d}}$. \end{obs}

\begin{proof} From Definition~\ref{defi:est},
\begin{eqnarray*}
	\est_{i} &=& \Psi_i + \sum\limits_{(A_1,\ldots,A_d,w) \in \tuple_i} w\cdot m_o(\dsubset) \\
		 &=& \Psi_i + \sum\limits_{(A_1,\ldots,A_d,w) \in \tuple_i^{\leq \tau}} w\cdot m_o(\dsubset) + \sum\limits_{(A_1,\ldots,A_d,w) \in \tuple_i^{>\tau}} w\cdot m_o(\dsubset) 
\end{eqnarray*}

In Step $2$ of the algorithm, for each tuple $(\dsubset,w) \in \tuple_i^{\leq \tau}$, we determine the exact value $m_o(\dsubset)$, add $w \cdot m_o(\dsubset)$ to current $\Psi$ and remove the tuple from $\cal D$. Observe that
\begin{equation}
	\Psi_{i+1} -\Psi_i = \sum\limits_{(\dsubset,w)\in \tuple_i^{\leq \tau}}w\cdot m_o(\dsubset) \label{eqn:2}
\end{equation}

If $\tuple_i^{>\tau}$ is empty, we go to Step $1$ to report the output. Observe that in that case $\est_{i+1}=\est_i$, and we are done. If $\tuple_i^{>\tau}$ is non-empty, then we go to either Step $3$ or Step $4$ depending on whether the number of tuples in $\cal D$ is at most $\cN$ or more than $\cN$, respectively, where $\cN=\kappa_d\frac{\log ^{4d}n}{\eps^2}$.

	{\bf Case 1: (Go to Step $3$)} 
Note that for each tuple $(\dsubset,w)$ in $\cal D$ we have $\tuple_{i}^{\geq \tau}$. We apply the sparsification step (Step $3$) for each tuple. For each tuple $(\dsubset,w)$, we add a set of tuples $\cZ$ by removing $(\dsubset,w)$ from $\cal D$. By Lemma~\ref{lem:sparse}, we have the following with probability $1-\frac{1}{n^{4\theta -2d}}$.
	$$\size{k\sum\limits_{(B_1,\ldots,B_d,4w) \in \cZ}m_o(B_1,\ldots,B_d) - m_o(\dsubset)} \leq k 2^{2d}\theta^d   \sqrt{ d! m_o(\dsubset) \log^d n}.$$ 

	Now using $m_o(\dsubset) \geq \tau=\frac{k^2\cdot 4^{2d}\theta^{2d} \cdot 16d^2 \cdot d! \log ^{d+2} n}{\eps^2}~\mbox{and}~k=4$ and taking $\theta=2d$, we have Equation \ref{eqn:3} with probability $1-\frac{1}{n^{6d}}$.

\begin{equation}
	\size{\sum\limits_{(B_1,\ldots,B_d,4w) \in \cZ}4w \cdot m_o(B_1,\ldots,B_d) - w\cdot m_o(\dsubset)} \leq  \frac{\eps}{4d \log n}   \cdot wm_o(\dsubset) \label{eqn:3}
\end{equation}

	Since we are in Step $3$, there are at most $\cN=\kappa_d \frac{\log ^{4d} n}{\eps^2}$ many tuples in $\tuple_i^{>\tau}$. As $\eps > \left(\frac{\log^{5d+5}}{n^d}\right)^{1/4}$, the probability that Equation \ref{eqn:3} holds for each tuple in $\tuple_i^{>\tau}$ is at least $1-\frac{1}{n^{5d}}$.

By Definition~\ref{defi:est},
\begin{eqnarray*}
	\est_{i+1} &=& \Psi_{i+1} + \sum\limits_{(B_1,\ldots,B_d,w') \in \tuple_{i+1}} w' \cdot m_o(B_1,\ldots,B_d) 
\end{eqnarray*}

	Using Equations \ref{eqn:2} and \ref{eqn:3}, we can show that $\est_{i+1}$ is an $(1+\lambda)$-approximation to $\est_i$, where $\lambda=\frac{\eps}{4d \log n}$, and the probability of success is $1-\frac{1}{n^{5d}}$.

	{\bf Case $2$: (Go to Step $4$)} 
Here, we apply coarse estimation algorithm for each tuple $(\dsubset,w)$ present in $\cal D$ to find $\hat{E}$ such that $\frac{m_o(\dsubset)}{\alpha} \leq \hat{E} \leq \alpha m_o(\dsubset)$ as described in Step-4. By Lemma~\ref{lem:coarse_main}, the probability of success of finding the required coarse estimation for a particular tuple, is at least $1-\frac{1}{n^{8d}}.$ By Observation~\ref{obs:num_tuples}, we have at most $4^d \cN =\frac{\kappa_d 4^d \log ^{4d} n}{\eps^2}$ many tuples at any instance of the algorithm. Hence, as $\eps > \left( \frac{\log ^{5d+5} n}{n^d}\right)^{1/4}$, the probability that we have the desired coarse estimation for all tuples present in $\cal D$, is at least $1-\frac{1}{n^{6d}}$.
 We have $r > \cN=\kappa_d \frac{\log ^{4d} n}{\eps^2}$ many tuples in $\cal D$. Under the conditional space that we have the desired coarse estimation for all tuples present in $\cal D$, we apply the algorithm {\sc Alg} corresponding to Lemma~\ref{lem:importance-app}. In doing so, we get $r' \leq \cN$ many tuples, as described in the Step-4, with probability $1-\frac{1}{n^{6d}}$. Observe that $\tuple_{i+1}$ is the set of $r'$ tuples returned by {\sc Alg} satisfying
 \begin{equation}\label{eqn:4}
 \size{ \sum\limits_{(B_1,\ldots,B_d,w')} w'\cdot m_o(B_1,\ldots,B_d) - S} 
  \leq \lambda S,
 \end{equation}
where $\lambda = \frac{\eps}{4d \log n}$ and $S=\sum\limits_{(\dsubset,w) \in \tuple_i^{>\tau}} w \cdot m_o(\dsubset)$. Now, by Definition~\ref{defi:est},
  \begin{eqnarray*}
  \est_{i+1} &=& \Psi_{i+1} + \sum\limits_{(B_1,\ldots,B_d,w') \in \tuple_{i+1}} w' \cdot m_o(B_1,\ldots,B_d) 
  \end{eqnarray*}
 Using Equations~\ref{eqn:2} and~\ref{eqn:4}, we can show that $\est_{i+1}$ is an $(1+\lambda)$-approximation to $\est_i$ and the probability of success  is $1-\left(\frac{1}{n^{6d}}+\frac{1}{n^{6d}}\right)\geq 1-\frac{1}{n^{6d-1}}$.
\end{proof}

\begin{obs}\label{lem:iter2} Let there be at least one tuple $(A_1,\ldots,A_d,w)$ in $\cal D$ after the $i$-th iteration such that $m_o(\dsubset) > \tau$ for any integer $i>0$. Then, $\act_{{i+2}} \leq \frac{\act_{i}}{2}$ with probability at least $1-\frac{2}{n^{5d}}$. \end{obs}

\begin{proof}[Proof of Observation~\ref{lem:iter2}]
As there exists one tuple in $\tuple_i^{>\tau}$, our algorithm will not terminate in Step-2. It will determine the exact values of
$m_o(\dsubset)$ for each $(\dsubset,w)\in \tuple_i^{\leq \tau}$, and then will go to either Step-3 or Step-4 depending on the cardinality of $\tuple_i^{\leq \tau}$. By adapting the same approach as that in the proof of Observation~\ref{lem:iter1}, we can show that
\begin{itemize}
\item[(i)] In the $(i+1)$-th iteration, if our algorithm goes to Step-3, then $\act_{i+1} \leq \frac{\act_{i}}{2}$ with probability $1-\frac{1}{n^{5d}}$; and
\item[(ii)] In the $(i+1)$-th iteration, if our algorithm goes to Step-4, then $\act_{i+1} \leq {\act_{i}}$ with probability $1-\frac{1}{n^{6d-1}}$.
\end{itemize}
From the description of the algorithm, it is clear that we apply sparsification either in iteration $(i+1)$ or $(i+2)$. That is, either we do sparsification in both the iterations, or we do sparsification in one iteration and coarse estimation in the other iteration, or we do sparsification in $(i+1)$-th iteration and termination of the algorithm after executing Step-2 in $(i+2)$-th iteration. Observe that in the last case, that is, if we terminate in $(i+2)$-th iteration, then $\act_{i+2}=0 \leq \frac{\act_{i}}{2}$. In other two cases, by (i) and (ii), we have $\act_{i+2}\leq \frac{\act_i}{2}$ with probability at least $1-\frac{2}{n^{5d}}$.
\end{proof}

Now, we are ready to prove Lemma~\ref{lem:prob}.

\begin{proof}[Proof of Lemma~\ref{lem:prob}] Let $i^*$ be the largest integer such that there exists at least one tuple $(A_1,\ldots,A_d,w)$ in the data structure $\cal D$ in the $i^*$-th iteration such that $m_o(\dsubset) > \tau$. That is $\act_{i^*} > \tau$. For ease of analysis, let us define the two following events:
\begin{description}
\item[$\cE_1$:] $i^* \leq 2d \log n$. 
\item[$\cE_2$:] $\est_{i^*}$ is an $(1 \pm \eps)$-approximation to $m_o(\cH)$.
\end{description}
Using the fact $\act_0=m_o(\cH) \leq n^d$ along with the Observation~\ref{lem:iter2}, we have $i^*\leq 2d\log n$ with probability at least $1-{2d \log n}\frac{2}{n^{5d}}$. That is $\pr(\cE_1) \geq 1-\frac{4d \log n}{n^{5d}}$.

Now let us condition on the event $\cE_1$. By the definition of $i^*$, we do the following in the $(i^*+1)$-th iteration. In Step 2, for each tuple $(\dsubset,w)$ present in $\cal D$, we determine $m_0(\dsubset)$ exactly, add it to $\Psi$ and remove $(\dsubset,w)$ from $\cal D$. Observe that $\act_{i^*+1}=0$, that is, $\est_{i^*+1}= \Psi_{i^* +1}=\est_{i^*}$. Since there are no tuples left in $\cal D$, we go to Step $1$. At the start of the $(i^*+2)$-th iteration, we report $\Psi_{i^* +1}=\est_{i^*}$ as the output. Using Observation~\ref{lem:iter1}, $\est_{i^*}$ is an ${(1 \pm \lambda)}^{i^*}$-approximation to $\est_0$ with probability at least $1-\frac{2d \log n}{n^{5d}}$. As $\est_0=m_o(\cH)$, $\lambda=\frac{\eps}{4d \log n}$, and $\cE_1$ has occurred, we have $\est_{i^*}$ is an $(1 \pm \eps)$-approximation to $m_o(\cH)$ with probability at least $1-\frac{2d\log n}{n^{3d+1}}$. That is $\pr(\cE_2~|~\cE_1) \geq 1-\frac{2d \log n}{n^{5d}}$.

Now, we analyze the query complexity of the algorithm on the conditional space that the events $\cE_1$ and $\cE_2$ have occurred. By the description of the algorithm, we make $\Oh_d\left(\frac{\log^{d+3}n}{\eps^2}\right)$ many \gonepis queries per tuple in Step $2$, and $\Oh_d(\log ^{d+1}n)$ many \gtwopis queries per tuple in Step $4$. Using Observation~\ref{obs:num_tuples}, there can be $\Oh_d\left( \frac{\log ^{4d} n}{\eps^2}\right)$ many tuples present in any iteration. Recall that the number of iterations is $i^*+2$, that is, $\Oh_d(\log n)$. Since $i^* \leq 2d \log n$, the query complexity of our algorithm is $\Oh_d \left(\log n \cdot \frac{\log ^{4d} n}{\eps^2} \cdot \left( \frac{\log^{d+3}n}{\eps^2} +\log^{d+1} n\right) \right)=\Oh_d\left( \frac{\log ^{5d+4} n}{\eps^4}\right)$, where each query is either a \gonepis or a \gtwopis query.

Now we compute the probability of success of our algorithm. Observe that
\begin{center}
$\pr(\mbox{{\sc Success}}) \geq \pr(\cE_1 \cap \cE_2)
= \pr(\cE_1) \cdot \pr(\cE_2~|~\cE_1) 
\geq \left(1-\frac{4d \log n}{n^{5d}}\right) \cdot \left(1-\frac{2d \log n}{n^{5d}}\right) 
 \geq 1-\frac{1}{n^{4d}}
$
\end{center}

\end{proof}



\bibliographystyle{alpha}
\bibliography{reference}
\newpage

\appendix
\section{Some probability results} \label{sec:prelim}

\begin{pro} \label{pro:exp} Let $X$ be a random variable. Then $\E[X] \leq \sqrt{\E[X^2]}$. \end{pro}

\begin{lem}[Theorem 7.1 in ~\cite{DubhashiP09}] \label{lem:dp} Let $f$ be a function of $n$ random variables $X_1,\ldots,X_n$ such that 
\begin{itemize}
	\item[(i)] Each $X_i$ takes values from a set $A_i$,
	\item[(ii)] $\E[f]$ is bounded, i.e., $0 \leq \E[f] \leq M$,
	\item[(iii)] $\cB$ be any event satisfying the following for each $i \in [n]$.
$$ \size{\E[f ~|~X_1,\dots,X_{i-1},X_{i}=a_i,\cB^c] - \E[f ~|~X_1,\dots,X_{i-1},X_{i}=a'_i,\cB^c] }\leq c_i$$
\end{itemize}  

Then for any $\delta \geq 0$, 
$$\pr\left(\size{f - \E[f]} > \delta + M\pr(\cB) \right) \leq \exp{\left(-{\delta^2}/{\sum\limits_{i=1}^{n} c_i^2}\right)} + \pr(\cB).$$
\end{lem}

\begin{lem}[\cite{DubhashiP09}(Hoeffding's inequality)]
\label{lem:hoeff_inq}
Let $X_1,\ldots,X_n$ be $n$ independent random variables such that $X_i \in [a_i,b_i]$. Then for $X=\sum\limits_{i=1}^n X_i$, the following is true for any $\delta >0$.
$$\pr \left( \size{X - \E[X]} \geq \delta \right) \leq 2 \cdot \exp{\left(-{2\delta ^2}/{\sum\limits_{i=1}^n (b_i - a_i)^2}\right)}.$$
\end{lem}

\begin{lem}[\cite{DubhashiP09}(Chernoff-Hoeffding bound)]
\label{lem:cher_bound1}
Let $X_1, \ldots, X_n$ be independent random variables such that $X_i \in [0,1]$. For $X=\sum\limits_{i=1}^n X_i$ and $\mu=\E[X]$, the followings hold for any $0\leq \delta \leq 1$.
$$ \pr(\size{X-\mu} \geq \delta\mu) \leq 2\exp{\left(-\mu \delta^2/3\right)}$$

\end{lem}
\begin{lem}[\cite{DubhashiP09}(Chernoff-Hoeffding bound)]
\label{lem:cher_bound2}
Let $X_1, \ldots, X_n$ be independent random variables such that $X_i \in [0,1]$. For $X=\sum\limits_{i=1}^n X_i$ and $\mu_l \leq \E[X] \leq \mu_h$, the followings hold for any $\delta >0$.
\begin{itemize}
\item[(i)] $\pr \left( X > \mu_h + \delta \right) \leq \exp{\left(-2\delta^2/n\right)}$.
\item[(ii)] $\pr \left( X < \mu_l - \delta \right) \leq \exp{\left(-2\delta^2 / n\right)}$.
\end{itemize}

\end{lem}
\begin{lem}~\cite{BeameHRRS18}
\label{lem:importance1}
Let $(D_1,w_1,e_1),\ldots, (D_r,w_r,e_r)$ are the given structures and each $D_i$ has an associated weight
${c}(D_i)$ satisfying  
\begin{itemize}
\item[(i)] $w_i,e_i \geq 1, \forall i \in [r]$;
\item[(ii)] $\frac{e_i}{\rho} \leq c(D_i) \leq e_i \rho$ for some $\rho >0$ and all $i \in [r]$; and
\item[(iii)] $\sum\limits_{i=1}^r {w_i\cdot c(D_i)} \leq M$.
\end{itemize}
Note that the exact values $c(D_i)$'s are not known to us. Then there exists an algorithm that finds 
 $(D'_1,w'_1,e'_1),\ldots, (D'_s,w'_s,e'_s)$ such that all of the above three conditions hold and 
 $\size{\sum\limits_{i=1}^t {w'_i\cdot c(D'_i)} - \sum\limits_{i=1}^r {w_i\cdot c(D_i)}} \leq \lambda S$ with probability $1-\delta$; where  $S=\sum\limits_{i=1}^r {w_i\cdot c(D_i)}$ and $\lambda, \delta >0$. The time complexity of the algorithm is $\cO(r)$ and 
  $s=\cO\left( \frac{\rho^4 \log M \left(\log \log M + \log \frac{1}{\delta}\right)}{\lambda^2}\right).$
\end{lem}  

\section{Oracle definitions} \label{sec:oracle-def}

\begin{defi}
{Independent set oracle (\bis)~\cite{BeameHRRS18}:} Given a subset $X$ of the vertex set of a graph $G(V,E)$ as input, the oracle returns {\sc Yes} if and only if $m(X)\neq 0$, where $m(X)$ denotes the number of edges in $G$ having both the vertices in $X$.
\end{defi}

\begin{defi}
{Bipartite independent set oracle (\bis)~\cite{BeameHRRS18}:} Given two disjoint subsets $A,B$ of the vertex set of a graph $G(V,E)$ as input, the oracle returns {\sc Yes} if and only if $m(A,B) \neq 0$, where $m(A,B)$ denotes the number of edges having exactly one vertex in both $A$ and $B$.
\end{defi}

\begin{defi}
{Tripartite independent set oracle (\tis)~\cite{}:} Given three pairwise disjoint subsets $A,B,C$ of the vertex set of a graph $G(V,E)$ as input, the oracle returns {\sc Yes} if and only if $t(A,B,C)$, where $t(A,B,C)$ denotes the number of triangles in $G$ having exactly one vertex in each  of $A$, $B$ and $C$.
\end{defi}

\begin{defi}
{Generalized $d$-partite independent set oracle (\gpis)~\cite{BishnuGKM018}:} Given $d$ pairwise disjoint subsets of vertices $\dsubset \subseteq U(\cH)$ of a hypergraph $\cH$ as input, \gpis{} query oracle answers {\sc Yes} if and only if $m(\dsubset) \neq 0$, where $m(\dsubset)$ denotes the number of hyperedges in $\cH$ having exactly one vertex in each $A_i$, $\forall i \in \{1,2,\ldots,d\}$. \end{defi}

\begin{defi}
{\gonepis oracle:} Given $s$ pairwise disjoint subsets of vertices $A_1,\ldots,A_s \subseteq U(\cH)$ of a hypergraph $\cH$ and $a_1,\ldots,a_s \in [d]$ such that $\sum_{i=1}^{s}a_i =d$, \gonepis query oracle on input $A_1^{[a_1]},A_2^{[a_2]},\cdots,A_s^{[a_s]}$ answers {\sc Yes} if and only if $m(\ssubset) \neq 0$.
\end{defi}

\begin{defi}
 {\gtwopis oracle:} Given any $d$ subsets of vertices $\dsubset \subseteq U(\cH)$ of a hypergraph $\cH$, \gtwopis query oracle on input $A_1,\ldots,A_d$ answers  {\sc Yes}  if and only if $m(\dsubset) \neq 0$.
\end{defi}

\section{\gpis query oracle and its polylogarithmic equivalents} \label{append:gpis}
\noindent Notice that the \gpis query oracle takes as input $d$ pairwise disjoint subsets of vertices. We
now define two other query oracles \gpis{}$ _1$ and \gpis{}$_2$ that are not as restrictive as \gpis in terms of admitting disjoint sets of vertices. We show shortly that both these oracles can be simulated by making polylogarithmic queries to \gpis with high probability. \gpis{}$ _1$ and \gpis{}$_2$ oracles will be used for ease of exposition. 
\begin{description}

	\item[($\mbox{{\sc GPIS}}_1$)] Given $s$ pairwise disjoint subsets of vertices $A_1,\ldots,A_s \subseteq U(\cH)$ of a hypergraph $\cH$ and $a_1,\ldots,a_s \in [d]$ such that $\sum\limits_{i=1}^{s}a_i =d$, \gonepis query oracle on input $A_1^{[a_1]},A_2^{[a_2]},\cdots,A_s^{[a_s]}$ answers {\sc Yes} if and only if $m(\ssubset) \neq 0$.

\item[($\mbox{{\sc GPIS}}_2$)] Given any $d$ subsets of vertices $\dsubset \subseteq U(\cH)$ of a hypergraph $\cH$, \gtwopis query oracle on input $A_1,\ldots,A_d$ answers  {\sc Yes}  if and only if $m(\dsubset) \neq 0$.
\end{description}
If from \gpis, we rule out the pairwise disjoint subset condition on $A_1, \ldots, A_d$, we get \gtwopis. In \gonepis, we allow multiple repetitions of the same set. 

From the above definitions, it is clear that a \gpis query can be simulated by a \gonepis or \gtwopis query. Through the following observations, we show how a \gonepis or a \gtwopis query can be simulated by polylogarithmic many \gpis queries.

\begin{obs}
\label{obs:append_queryoracles}
\begin{itemize}
\item[(i)] A \gonepis query can be simulated by using polylogarithmic \gpis queries with high probability.
\item[(ii)] A \gtwopis query can be simulated using $2^{\Oh(d^2)}$ \gonepis queries.
\item[(iii)] A \gtwopis query can be simulated using polylogarithmic \gpis queries with high probability.
\end{itemize}
\end{obs}

\begin{proof}
\begin{itemize}
	\item[(i)] Let the input of \gonepis query oracle be $\ssubset$ such that $a_i \in [d]~\forall i \in [s]$ and $\sum\limits_{i=1}^s a_i =d$. For each $i \in [s]$, we partition $A_i$ (only one copy of $A_i$, and not $a_i$ many copies of $A_i$) randomly into $a_i$ parts, let $\{B_i^j:j \in [a_i]\}$ be the resulting partition of $A_i$.   Then we make a \gpis query with input $B_1^{1},\ldots,B_1^{a_1},\ldots, B_s^{1},\ldots,B_s^{a_s}$. Note that 
$$
\cF(B_1^{1},\ldots,B_1^{a_1},\ldots, B_s^{1},\ldots,B_s^{a_s}) \subseteq \cF(\ssubset).
$$

So, if \gonepis outputs {\sc `No'} to query $\ssubset$, then the above \gpis query will also report {\sc `No'} as its answer. If \gonepis answers {\sc `Yes'}, then consider a particular hyperedge $F \in \cF(\ssubset)$. Observe that 
\begin{eqnarray*}
&& \pr(\mbox{\gpis oracle answers {\sc `Yes'}})\\
&\geq& \pr(\mbox{$F$ is present in $\cF(B_1^{1},\ldots,B_1^{a_1}, \ldots \ldots, B_s^{1},\ldots,B_s^{a_s})
$})\\
	&\geq& \prod\limits_{i=1}^s \frac{1}{a_i ^{a_i}} \\ 
	&\geq& \prod\limits_{i=1}^s \frac{1}{d ^{a_i}}~~~~~~~~~~(\because a_i \leq d~\mbox{for all}~i\in [d])\\
	 &=& \frac{1}{d^{d}}~~~~~~~~~~(\because \sum\limits_{i=1}^s a_i =d)
\end{eqnarray*}

We can boost up the success probability arbitrarily by repeating the above procedure polylogarithmic many times.
  
	\item[(ii)] Let the input to \gtwopis query oracle be $\dsubset$. Let us partition each set $A_i$ into at most $2^{d-1}-1$ subsets depending on $A_i$'s intersection with $A_j$'s for $j \neq i$. Let $\cP_i$ denote the corresponding partition of $A_i$, $i \in [d]$. Observe that for any $i \neq j$, if we take any $B_i \in \cP_i$ and $B_j \in \cP_j$, then either $B_i=B_j$ or $B_i \cap B_j = \emptyset$.
  
For each $(B_1,\ldots,B_d) \in \cP_1 \times \ldots \times \cP_d$, we make a \gonepis query with input $(B_1,\ldots,B_d)$. Total number of such \gonepis queries is at most $2^{\Oh(d^2)}$, and we report {\sc `Yes'} to the \gtwopis query if and only if at least one \gonepis query, out of the $2^{\Oh(d^2)}$ queries, reports {\sc `Yes'}.
 
 \item[(iii)] It follows from (i) and (ii).
\end{itemize}

\end{proof}

\end{document}